\documentclass[a4paper,USenglish,11pt]{article}
\usepackage{hyperref}

\usepackage[margin=1in]{geometry}
\usepackage[utf8]{inputenc}
\usepackage{algorithm}
\usepackage{algorithmicx}
\usepackage[noend]{algpseudocode}
\usepackage[noadjust]{cite}
\usepackage[OT4]{fontenc}
\usepackage{amsthm}
\usepackage{amssymb}
\usepackage{amsmath}
\usepackage{amsfonts}
\usepackage{todonotes}
\usepackage{authblk}
\usepackage{rotating}
\usepackage{xspace}
\usepackage[shortlabels]{enumitem}
\usepackage{verbatim}
\usepackage{mathrsfs}
\usepackage{framed}
\usepackage{thm-restate}
\usepackage[capitalise]{cleveref}
\usepackage{colortbl}
\usepackage{array,multirow}

\def\dd{\mathinner{.\,.}}
\newcommand{\cO}{\mathcal{O}}
\newcommand{\cOtilde}{\tilde{\cO}}
\newcommand{\eps}{\varepsilon}
\newcommand{\PT}{\mathsf{PT}}

\newcommand{\CUp}{\mathsf{C_{up}}}
\newcommand{\CDown}{\mathsf{C_{down}}}
\newcommand{\C}{\mathsf{C}}
\newcommand{\dUp}{\mathsf{d_{up}}}
\newcommand{\dDown}{\mathsf{d_{down}}}
\newcommand{\dUpDown}{\mathsf{d}}
\newcommand{\JUp}{\mathsf{J_{up}}}
\newcommand{\JDown}{\mathsf{J_{down}}}
\newcommand{\YL}{\mathsf{Y_{\ell}}}
\newcommand{\YR}{\mathsf{Y_{r}}}
\newcommand{\ZL}{\mathsf{Z_{\ell}}}
\newcommand{\ZR}{\mathsf{Z_{r}}}
\newcommand{\gen}{\textsf{gen}}
\newcommand{\val}{\textsf{val}}
\newcommand{\pathtree}{\textsf{path}}

\def\poly{\operatorname{poly}}
\newcommand{\lev}{\textsf{lev}}
\newcommand{\Lv}{\mathcal{L}}

\newcommand{\makestring}{\textsf{makestring}}

\newcommand{\lcp}{\textsf{lcp}}
\newcommand{\decompose}{\textsc{decompose}}

\theoremstyle{plain}
  \newtheorem{theorem}{Theorem}[section]
  \newtheorem{lemma}[theorem]{Lemma}  
  \newtheorem{corollary}[theorem]{Corollary} 
  \newtheorem{proposition}[theorem]{Proposition}

  \theoremstyle{definition}
  \newtheorem{definition}[theorem]{Definition}
  
  \newtheorem{example}[theorem]{Example}
  \newtheorem{remark}[theorem]{Remark}
  \newtheorem*{claim}{Claim}

\usepackage{etoolbox}
\makeatletter
\setbool{@fleqn}{false}
\makeatother

\title{Dynamic Longest Common Substring in Polylogarithmic Time}

\author[1]{Panagiotis Charalampopoulos}
\author[2]{Pawe{\l} Gawrychowski}
\author[2]{Karol Pokorski}

\affil[1]{King’s College London, London, UK\\
    \texttt{p.charalampopoulos@kcl.ac.uk}}
\affil[2]{Institute of Computer Science, University of Wrocław, Poland\\
	\texttt{$\{$gawry,pokorski$\}$@cs.uni.wroc.pl}}
	
\date{\vspace{-5ex}}

\newcommand{\FIGURE}[4]{
\begin{figure}[#1]
\begin{centering}
\includegraphics[width={#2}\textwidth]{#3.pdf}
\caption{#4}
\label{fig:#3}
\end{centering}
\end{figure}
}

\newcommand{\genproblem}[3]{
\begin{framed}
  \noindent
  \textbf{Problem:} \textsc{#1}

  \noindent
  \textbf{Input:} #2

  \noindent
  \textbf{Output:} #3
\end{framed}
}

\newcommand{\DSproblem}[3]{
\begin{framed}
  \noindent
  \textbf{Problem:} \textsc{#1}

  \noindent
  \textbf{Input:} #2

  \noindent
  \textbf{Query:} #3
\end{framed}
}

\newcommand{\Dynproblem}[4]{
\begin{framed}
  \noindent
  \textbf{Problem:} \textsc{#1}

  \noindent
  \textbf{Input:} #2

  \noindent
  \textbf{Update:} #3
  
  \noindent
  \textbf{Query:} #4
\end{framed}
}

\newcommand{\T}{\mathcal{T}}

\begin{document}
\maketitle

\thispagestyle{empty}

\begin{abstract}
The longest common substring problem consists in finding a longest string that appears as a (contiguous) substring of
two input strings. We consider the dynamic variant of this problem, in which we are to maintain two dynamic strings $S$
and $T$, each of length at most~$n$, that undergo edit operations, i.e., substitutions, insertions, and deletions of letters, in order to be able to return a longest common
substring after each edit.
Amir, Charalampopoulos, Pissis, and Radoszewski [Algorithmica 2020] presented a solution for this problem that requires $\tilde{\mathcal{O}}(n^{2/3})$ time per update.
This brought the challenge of determining whether there exists a solution
with polylogarithmic update time or we should expect a polynomial (conditional) lower bound.
We answer this question by designing an exponentially faster algorithm that processes each edit operation in amortized $\cO(\log^7 n)$ 
time with high probability.
Our solution relies on exploiting the local consistency of the parsing of a collection of dynamic strings due to
Gawrychowski, Karczmarz, Kociumaka, Łącki, and Sankowski [SODA 2018], and on maintaining two dynamic trees with labeled bicolored leaves.
We complement our solution with a lower bound of $\Omega(\log n/ \log\log n)$ for the update time of any
polynomial-size data structure that maintains an LCS of two dynamic strings,
and the same lower bound for the update time of any $\tilde{\mathcal{O}}(n)$-size data structure that maintains the LCS of a static and
a dynamic string. Both lower bounds hold even allowing amortization and randomization.
This requires extending Pătraşcu's reduction from the 
lopsided set disjointness problem to the butterfly reachability problem [SICOMP 2011].
\end{abstract}

\clearpage
\setcounter{page}{1}

\section{Introduction}

The well-known longest common substring (LCS) problem, formally stated below, was conjectured by Knuth to require
$\Omega(n \log n)$ time.  However, in his seminal paper that introduced suffix trees, Weiner showed how to solve it in
linear time (for constant alphabets)~\cite{DBLP:conf/focs/Weiner73}. 
For polynomially-bounded integer alphabets, the linear-time construction of suffix trees by Farach yielded an $\cO(n)$-time algorithm~\cite{DBLP:conf/focs/Farach97}.
A faster algorithm can be achieved in the word RAM model of computation when the alphabet is small~\cite{DBLP:conf/esa/Charalampopoulos21}.
Recently, the quantum computational complexity of the LCS problem was settled to be at most a polylogarithmic factor away from $n^{2/3}$~\cite{DBLP:conf/innovations/GallS22,DBLP:conf/soda/AkmalJ22,DBLP:journals/corr/abs-2211-15945}. 
Further, many different versions of this classical question have been 
considered, such as obtaining a tradeoff between the time and the working
space~\cite{DBLP:conf/cpm/StarikovskayaV13,DBLP:conf/esa/KociumakaSV14,DBLP:conf/cpm/Nun0KK20,DBLP:conf/cpm/BathieCS24}, or computing an approximate LCS under either
the Hamming or the edit distance (see~\cite{DBLP:journals/jcb/ThankachanAA16,DBLP:conf/cpm/Charalampopoulos18,DBLP:journals/corr/abs-1712-08573,DBLP:conf/esa/Charalampopoulos21} and references therein).

\genproblem{\textsc{Longest Common Substring}}
{Two strings $S$ and $T$ of length at most $n$ over an alphabet $\Sigma$.}
{A longest substring $X$ of $S$ that is a substring of $T$.}

We consider the dynamic version of the LCS problem, where the strings are updated and we are to report an LCS after each update.
That is, we return the length of an LCS and a pair of starting positions of its occurrences in the strings.
The allowed update operations are edit operations, i.e., substitutions, insertions, and deletions of single letters in either $S$ or $T$.
For a variant that we consider, we actually show that if only substitution operations are allowed, the problem admits more efficient solutions.

Dynamic problems on strings are of wide interest. Maybe the most basic question in this direction
is that of maintaining a dynamic text while enabling efficient pattern matching queries.
This is clearly motivated by the possible application in a text editor, where
the text is dynamic and the user may issue pattern matching queries.
A considerable amount of work has been carried out on this problem~\cite{Gu94,Fer97,FG98}.
The first structure achieving
polylogarithmic update time and optimal query time for this problem was designed by Sahinalp and Vishkin~\cite{SahinalpVishkin}.
Later, the update time was improved to $\cO(\log^2 n \log \log n \log^* n)$ at the cost of $\cO(\log n \log \log n)$
additional time per query by Alstrup, Brodal, and Rauhe~\cite{Alstrup00patternmatching}.
Recently, Gawrychowski et al.~\cite{ods} presented a data structure that requires $\cO(\log^2 n)$ time per
update and allows for time-optimal queries.
Another problem that has obtained significant attention is that of maintaining a dynamic collection of strings
so that equality queries or longest common prefix queries can be answered efficiently~\cite{st94,ksu97,Alstrup00patternmatching,ods,DBLP:conf/sosa/CardinalI21,DBLP:conf/sosa/AxiotisBBJNTW21,DBLP:conf/esa/LiptakM024}.
Other problems on strings that have been studied in the dynamic setting include
string alignment and (weighted) edit distance~\cite{DBLP:journals/jda/HyyroNI15,DBLP:conf/cpm/Charalampopoulos20a,DBLP:journals/corr/abs-2404-06401,DBLP:conf/focs/CassisKW23,DBLP:conf/focs/Kociumaka0S23,DBLP:conf/esa/BonehGK25},
longest increasing subsequence~\cite{DBLP:conf/stoc/MitzenmacherS20,DBLP:journals/corr/abs-2011-09761,DBLP:journals/corr/abs-2011-10874,DBLP:journals/corr/abs-2102-11797},
time warping distance~\cite{DBLP:conf/spire/NishiNIBT20,DBLP:conf/soda/BringmannFHKKR24},
approximate pattern matching~\cite{CGLS18,unified,DBLP:conf/cpm/CliffordGK0U22},
internal pattern matching~\cite{unified,DBLP:conf/stoc/KempaK22,DBLP:conf/spire/DuysterK24},
closest string~\cite{DBLP:journals/corr/abs-2205-00441},
maintaining the suffix array~\cite{DBLP:conf/isaac/AmirB20,DBLP:conf/stoc/KempaK22,DBLP:journals/corr/abs-2404-07510} and the LZ77 factorization~\cite{DBLP:conf/soda/Boneh0K26},
as well as maintaining repetitions~\cite{amir_et_al:LIPIcs:2019:11126} and
a longest palindromic substring~\cite{amir19,DBLP:journals/corr/abs-1906-09732}.
Additionally, dynamic (approximate) pattern matching for 2D-strings has been studied~\cite{DBLP:conf/spire/CliffordFSV16}.

As for the LCS problem itself, Amir, Charalampopoulos, Iliopoulos, Pissis, and Radoszewski \cite{amir17} initiated the study of this question in the dynamic
setting by considering the problem of constructing a data structure over two strings that returns the LCS after
a single edit operation in one of the strings. However, in their solution, after each edit operation, the string
is immediately reverted to its original version. Abedin, Hooshmand, Ganguly, and Thankachan~\cite{DBLP:conf/cpm/AbedinH0T18} improved the 
tradeoffs for this problem by designing a more efficient solution for the so-called heaviest induced ancestors problem.
Amir and Boneh~\cite{DBLP:conf/cpm/AmirB18} investigated some special cases of the \emph{partially dynamic LCS} problem
(in which one of the strings is assumed to be static); namely, the case where the static string is periodic and the case
where the substitutions in the dynamic string are substitutions with some letter $\# \not\in \Sigma$.
Finally, Amir et al.~\cite{amir19} presented the first algorithm for the \emph{fully dynamic LCS} problem (in
which both strings are subject to updates) that needs only sublinear time per edit operation,
namely $\cOtilde(n^{2/3})$ time.
As a stepping stone towards this result, they designed an algorithm for the partially dynamic LCS problem
that processes each edit operation in $\cOtilde(\sqrt{n})$ time.

For some natural dynamic problems, the best known bounds on the query and the update time are of the form
$\cO(n^{\alpha})$, where $n$ is the size of the input and $\alpha$ is some constant. Henzinger, Krinninger, Nanogkai, and Saranurak~\cite{HenzingerKNS15}
introduced the online Boolean matrix-vector multiplication conjecture that can be used to provide some justification for
the polynomial-time hardness of many such dynamic problems in a unified manner. This brings the question
of determining whether the bound on the update time in the dynamic LCS problem should be polynomial or subpolynomial.

We answer this question by significantly improving on the bounds presented by Amir et al.~\cite{amir19} and presenting
a solution for the fully dynamic LCS problem that handles each update in amortized $\cO(\log^7 n)$ time with
high probability, using $\cOtilde(n)$ space. As a warm-up, we present a (relatively simple) deterministic solution for the partially dynamic LCS problem that
handles each update in $\cO(\log n / \log\log n)$ time, using $\cOtilde(n)$ space.

After having determined that the complexity of fully dynamic LCS is polylogarithmic, the next natural question is
whether we can further improve the bound to polyloglogarithmic. By now, we have techniques that can be used to
not only distinguish between these two situations but (in some cases) also provide tight bounds. As a prime example,
static predecessor for a set of $n$ numbers from $[n^{2}]$ requires $\Omega(\log\log n)$ time for structures
of size $\cOtilde(n)$~\cite{patrascu06pred}, and dynamic connectivity for forests requires $\Omega(\log n)$
time~\cite{patrascu06loglb}, with both bounds being asymptotically tight. In some cases, seemingly similar problems
might have different complexities, as in the orthogonal range emptiness problem: Nekrich~\cite{DBLP:conf/compgeom/Nekrich07}
showed a data structure of size $\cO(n\log^{4}n)$ with $\cO(\log^2 \log n)$ query time for 3 dimensions, while 
for the same problem in 4 dimensions P{\v a}tra{\c s}cu showed that any polynomial-size data structure requires
$\Omega(\log n / \log \log n)$ query time~\cite{DBLP:journals/siamcomp/Patrascu11}.

For the case where the allowed updates are edit operations, 
we show that the partially dynamic LCS problem on strings of length $n$ is already hard:
we provide an unconditional $\Omega(\log n/ \log \log n)$
lower bound for the update time of any polynomial-size data structure.
Then, we focus on the more challenging case where only substitutions are allowed.
For the partially dynamic LCS problem on strings of length $n$, we provide an unconditional $\Omega(\log n/ \log \log n)$
lower bound for the update time of any data structure of size $\cOtilde(n)$.
For the fully dynamic LCS problem, we are able to show an unconditional $\Omega(\log n/ \log \log n)$ lower
bound for the update time of any polynomial-size data structure.
The latter two lower bounds hold even when both amortization and 
Las Vegas randomization are allowed.

Finally, we demonstrate that the difference in the allowed space between the above lower bounds is indeed necessary.
To this end, we show that when we only allow substitutions, partially dynamic LCS admits an $\cO(n^{1+\epsilon})$-space, $\cO(\log\log n)$-update time solution, for any constant $\epsilon>0$.

See \cref{tab:results} for a comprehensive overview of our results.

\renewcommand{\arraystretch}{1.7}
\begin{table}[htpb!]
\begin{centering}
\begin{tabular}{|c|c|c|c|c|}
\hline
\textbf{Variant} & \textbf{Operations} & \textbf{Space} & \textbf{Update time} &  \\
\hline
\multirow{7}{*}{Partially dynamic} & \multirow{3}{*}{Substitutions} & $\cOtilde(n)$ & $\Omega\left(\frac{\log n}{\log \log n}\right)$ & \cref{thm:lbpartially} \\
 & & $\cOtilde(n)$ & $\cO(\frac{\log n}{\log \log n})$ & \cref{thm:onesided}\\
 & & $\cO(n^{1+\epsilon})$ & $\cO(\log \log n)$ & \cref{cor:onesided}\\
\cline{2-5}
 & \multirow{2}{*}{Edit operations} & $\cO(\poly n)$ & $\Omega\left(\frac{\log n}{\log \log n}\right)$ & \cref{lb:indels} \\
 & & $\cOtilde(n)$ & $\cO\left(\frac{\log n}{\log \log n}\right)$ & \cref{thm:onesided_edit}\\
\hline
\multirow{2}{*}{Fully dynamic} & Substitutions & $\cO(\poly n)$ & $\Omega\left(\frac{\log n}{\log \log n}\right)$ & \cref{thm:lb}\\
\cline{2-5}
& Edit operations & $\cOtilde(n)$ & $\cO(\log^7 n)$ & \cref{thm:fully}\\
\hline
\end{tabular}
\caption{A summary of update time complexities for variants of the dynamic LCS problem. 
Parameter $\epsilon$ is a predefined arbitrarily small positive constant.
Note that several of the stated upper and lower bounds allow for amortization and/or randomization; for each of the bounds, this is explicitly stated in the corresponding section.\label{tab:results}}
\end{centering}
\end{table}

\paragraph{Techniques and roadmap.}

We first consider the partially dynamic version of the problem where updates are only allowed in one of the strings, say $S$,
in Section~\ref{sec:partially}.
This problem is easier as we can use the static string $T$ as a reference point.
We maintain a partition of $S$ into blocks (i.e., substrings of $S$
whose concatenation equals $S$), such that each block is a substring
of $T$, but the concatenation of any two consecutive blocks is not.
This is similar to the approach of~\cite{DBLP:journals/talg/AmirLLS07} and other works that consider one dynamic and one static string.
The improvement upon the $\cOtilde(\sqrt{n})$-time algorithm presented in~\cite{amir19} comes exactly from imposing the aforementioned maximality property, which guarantees that the sought LCS is a substring of the concatenation of at most three consecutive blocks and contains the first letter of one of these blocks.
The latter property allows us to anchor the LCS in $S$.
Upon an update, we can maintain the block decomposition, by updating a constant number of blocks.
It then suffices to show how to efficiently compute the longest substring of $T$ that contains the first letter of a given block.
We reduce this problem to answering a \emph{heaviest induced ancestors} (HIA) query.
This reduction was also presented in~\cite{amir17,DBLP:conf/cpm/AbedinH0T18}, but we describe the details to help the
reader develop intuition towards understanding the
more involved solution of fully dynamic LCS.

In Section~\ref{sec:fully}, we move to the fully dynamic LCS problem. We try to anchor the LCS in both strings as follows. 
For each of the strings $S$ and $T$ we show how to maintain, in $\log^{\cO(1)} n$ time, a collection of pairs of adjacent fragments (e.g.~$(S[i \dd j-1],S[j \dd k])$), denoted by $J_S$ for $S$ and $J_T$ for $T$, with the following property. For any common substring $X$ of $S$ and $T$ there exists a partition $X=X_{\ell}X_r$ for which there exists a pair $(U_{\ell},U_r) \in J_S$ and a pair $(V_{\ell},V_r) \in J_T$ such that $X_{\ell}$ is a suffix of both $U_{\ell}$ and $V_{\ell}$, while $X_r$ is a prefix of both $U_r$ and $V_r$.
We can maintain this collection by exploiting the properties of the locally consistent parsing previously used for maintaining a dynamic collection of strings~\cite{ods}.
We maintain tries for fragments in the collections $J_S$ and $J_T$, and reduce the dynamic LCS problem to a problem on dynamic bicolored trees, which we solve by using dynamic heavy-light decompositions and 2D range trees.

The lower bound for the case where edit operations are allowed follows from a variant of the dynamic partial sums problem~\cite{10.1145/73007.73040,DBLP:journals/siamcomp/HusfeldtR03}.
The lower bounds for the case where only substitution operations are allowed are much more involved.
For those, we first show a reduction from the problem of answering reachability queries in butterfly graphs that was considered in the seminal paper of Pătraşcu~\cite{DBLP:journals/siamcomp/Patrascu11} to the HIA problem.
However, both of these are data structure problems.
For the lower bound for the fully dynamic LCS problem, we overcome this by using as an intermediate step a reduction to a (restricted) dynamic variant of the HIA problem.
In fact, in order to show that these lower bounds hold even when Las Vegas randomization is allowed,
we slightly generalize Pătraşcu's reduction from the lopsided set disjointness problem to the butterfly reachability problem.
See~\cref{sec:lowerbound}.
 
\section{Preliminaries}

We use $[n]$ to denote the set $\{1, 2, \dots, n\}$.
Let $S=S[1]S[2]\cdots S[n]$ be a \emph{string} of length $|S|=n$ over an integer alphabet $\Sigma$. For two positions $i$ and $j$ on $S$, we denote by $S[i\dd j]=S[i]\cdots S[j]$ the \emph{fragment} of $S$ that starts at position $i$ and ends at position $j$ (it is the empty string $\varepsilon$ if $j<i$).
A string $Y$, of length $m$ with $0<m\leq n$, is a \emph{substring} of $S$ if there exists a position $i$ in $S$ such that $Y=S[i \dd i+m-1]$.
The prefix of $S$ ending at the $i$-th letter of $S$ is denoted by $S[\dd i]$ and the suffix of $S$ starting at the $i$-th letter of $S$ is denoted by $S[i \dd]$.
The reverse string of $S$ is denoted by $S^R$.
The concatenation of strings $S$ and $T$ is denoted by $ST$, and the concatenation of $k$ copies of string $S$ is denoted by $S^k$.
By $\lcp(S, T)$ we denote the length of the longest common prefix of
strings $S$ and $T$.

The \emph{trie} of a collection $C=\{S_1, S_2, \dots, S_k\}$ of strings is
the smallest rooted tree with edges labeled by single letters,
such that
every string $S$ that is a prefix of some string in $C$ is represented by some node $v$,
where the concatenation of the labels of the edges of the root-to-$v$ path,
the \emph{path-label} of $v$, is equal to $S$.
The compacted trie of $C$ is obtained by contracting maximal paths consisting of nodes with one child to an edge labeled by the concatenation of the labels of the edges of the path.
Usually, the label of the new edge is stored as the start/end indices of the corresponding fragment of some $S_{i}$.
The \emph{suffix tree} of a string $T$ is the compacted trie of all suffixes of $T \$$ where $\$$ is a letter smaller than all letters of the alphabet $\Sigma$.
It can be constructed in $\cO(|T|)$ time for linear-time sortable alphabets~\cite{DBLP:conf/focs/Farach97}.
For a node $u$ in a (compacted) trie, we define its \emph{depth} as the number
of edges on the path from the root to $u$.
Analogously, we define the \emph{string-depth} of $u$ as the total length of labels
along the path from the root to $u$.

We say that a tree is weighted if there is a weight $w(u)$ associated with each node $u$ of the tree, such that weights along the root-to-leaf paths are increasing, i.e., for any node $u$ other than the root, $w(u) > w(\operatorname{parent}(u))$. Further, we say that a tree is labeled if each of its leaves is given a distinct label.

\begin{definition}
  For rooted, weighted, labeled trees $\T_1$ and $\T_2$, two nodes $u \in \T_1$ and $v \in \T_2$,
   are \emph{induced} (by a label $\ell$) if and only if there are
  leaves $x$ and $y$ with the same label $\ell$, such that $x$ is a descendant of $u$
  and $y$ is a descendant of $v$.
\end{definition}

\DSproblem{\textsc{Heaviest Induced Ancestors}}
{Two rooted, weighted, labeled trees $\T_1$ and $\T_2$ of total size $n$.}
{Given a pair of nodes $u \in \T_1$ and $v \in \T_2$, return a pair of nodes $u', v'$ such that $u'$ is ancestor of $u$, $v'$ is ancestor of
$v$, $u'$ and  $v'$ are induced and they have the largest total combined weight $w(u') + w(v')$.}

This problem was introduced in~\cite{DBLP:conf/cccg/GagieGN13}, with the last advances made in~\cite{charalampopoulos2023optimal}.
The next lemma encapsulates one of the known trade-offs.

\begin{lemma}[\cite{charalampopoulos2023optimal}]\label{lem:hia}
There is a data structure for the \textsc{Heaviest Induced Ancestors} problem, that can be built in $\cOtilde(n)$ time and answers queries in $\cO(\frac{\log n}{\log\log n})$ time.
\end{lemma}

\section{Partially Dynamic LCS}
\label{sec:partially}

In this section, we describe an algorithm for solving the partially dynamic
variant of the LCS problem, where updates are only allowed on one of the strings, say $S$, while
$T$ is given in advance and is not subject to change.
We first consider the case where only substitution operations are allowed and then
extend our solution for the case where insertions and deletions are also allowed.

For simplicity, we assume that $S$ is initially equal to
$\texttt{\$}^{|S|}$, for $\texttt{\$} \not \in \Sigma$. We can always obtain any
other initial $S$ by performing an appropriate sequence of updates in the beginning.
After these initial updates, let us assume that all the letters of $S$ throughout
the execution of the algorithm occur at least once in $T$;
we waive this assumption later.

\begin{definition}\label{def:blockdec}
  A \emph{block decomposition} of string $S$ with respect to string $T$ is
  a sequence of strings $(s_1, s_2, \dots, s_k)$ such that
  $S = s_1s_2 \dots s_k$ and every $s_i$ is a fragment of $T$.
  An element of the sequence is called a \emph{block} of the decomposition.
  A decomposition is \emph{maximal} if and only if
  $s_is_{i+1}$ is not a substring of $T$ for every $i \in [k-1]$.
 \end{definition}

Maximal block decompositions are not necessarily
unique and may have different lengths, but all admit the following useful property.

\begin{lemma}
  For any maximal block decomposition of $S$ with respect to $T$, any fragment of $S$ that
  occurs in $T$ is contained in at most three consecutive blocks.
  Furthermore, any occurrence of an LCS of $S$ and $T$ in $S$ must contain the first letter of some block.
  \label{lem:lcs_three_blocks}
\end{lemma}
\begin{proof}
  We prove the first claim by contradiction. If $(s_1, s_2, \dots, s_k)$ is a maximal block decomposition
  of $S$ with respect to $T$ and a fragment of $S$ that occurs in $T$ spans at least four consecutive blocks
  $s_i, s_{i+1}, s_{i+2}, \dots, s_j$, then $s_{i+1}s_{i+2}$ is a substring of $T$, a contradiction.
  
  As for the second claim, it is enough to observe that if an occurrence of an LCS in $S$ starts in some other than the first position of a block $b$, then it must contain the first letter of the next block, as otherwise its length would be smaller than the length of block $b$,
  which is a common substring of $S$ and $T$.
\end{proof}

We will show that an update in $S$ can be processed by considering a constant number of blocks in a
maximal block decomposition of $S$ with respect to $T$.
We first summarize the basic building block needed for efficiently maintaining such a maximal block decomposition.

\begin{lemma}\label{lem:rangemerging}
Let $T$ be a string of length at most $n$. In $\cO(n \log^2 n)$-time, we can construct an $\cO(n\log n)$-space
data structure that, given two fragments $U$ and $V$ of $T$, can compute a longest fragment of $T$ that is equal to a prefix of
$UV$ in $\cO(\log \log n)$ time.
\end{lemma}
\begin{proof}
We build a weighted ancestor queries structure over the suffix tree of $T$. 
A weighted ancestor query $(\ell,u)$ on a (weighted) tree $\mathcal{T}$, asks for the deepest ancestor of $u$ with weight at most $\ell$.
Such queries can be answered in $\cO(\log \log n)$ time after an $\cO(n)$-time preprocessing of~$\mathcal{T}$ if all weights
are polynomial in $n$~\cite{FarachM96}, as is the case for suffix trees with the weight of each node being its string-depth.
We also build a data structure for answering unrooted LCP queries
over the suffix tree of $T$.
In our setting, such queries can be defined as follows: given nodes $u$ and~$v$ of the suffix tree of $T$,
we want to compute the (implicit or explicit) node where the search for the path-label of $v$ starting from node $u$ ends.
Cole, Gottlieb, and Lewenstein~\cite{DBLP:conf/stoc/ColeGL04} showed how to construct in $\cO(n\log^{2}n)$ time a data structure of size $\cO(n\log n)$
that answers unrooted LCP queries in $\cO(\log\log n)$ time.
With these data structures at hand, the longest prefix of $UV$ that is a fragment of $T$ can be computed as follows.
First, we retrieve the nodes of the suffix tree of $T$ corresponding to $U$ and $V$ using weighted ancestor queries in $\cO(\log\log n)$ time.
In more detail, if $U=T[i \dd j]$ then we access the leaf of the suffix tree corresponding to $T[i \dd ]$ and access its ancestor
at string-depth $|U|$, and similarly for $V$. Second, we ask an unrooted LCP query to obtain the node corresponding to the sought prefix of $UV$.
\end{proof}

\begin{lemma}\label{lem:blockupd}
A maximal block decomposition of a dynamic string $S$, with respect to a static string $T$,
can be maintained in $\cO(\log\log n)$ time per substitution operation with a data structure of size $\cO(n\log n)$ that can be
constructed in $\cO(n\log^2 n)$ time. 
\end{lemma}
\begin{proof}
We keep the blocks on a doubly-linked list and we store the starting positions of blocks in an $\cO(n)$-size
predecessor/successor data structure over $[n]$ that supports $\cO(\log \log n)$-time queries and updates~\cite{vanemdeboas1975}. 
This allows us to navigate in the structure of blocks, and in particular to be able to compute the block in which the edit occurred and its neighbors.

Suppose that we have a maximal block decomposition $B=(s_1, \dots, s_k)$ of $S$ with respect to~$T$.
Let us consider a substitution operation with letter $y$ at position $t$ of block $s_i$, and
  let $s_i = s_i^{\textrm{l}}x s_i^{\textrm{r}}$, with $|s_i^{\textrm{l}}x|=t$.
  Consider a block decomposition
  $B'=(s_1, s_2, \dots, s_{i-1}, s_i^{\textrm{l}}, y, s_i^{\textrm{r}}, s_{i+1}, \dots, s_k)$
  of string $S'$ after the update.
  Note that both $s_i^{\textrm{l}}$ and $s_i^{\textrm{r}}$ may be empty.
  This block decomposition does not need to be maximal. However, since $B$ is a maximal block decomposition
  of $S$, none of the strings $s_1s_2$, $s_2s_3$, $\ldots$, $s_{i-2}s_{i-1}$,
  $s_{i+1}s_{i+2}$, $s_{i+2}s_{i+3}$, $\ldots$, $s_{k-1}s_k$
  occurs in $T$. Thus, given $B'$, we repeatedly merge any two consecutive blocks from
  $(s_{i-1}, s_i^{\textrm{l}}, y, s_i^{\textrm{r}}, s_{i+1})$ whose concatenation is a substring of $T$
 into one, until this is no longer possible.
 We have at most four merges before obtaining a maximal block decomposition $B'$ of string $S'$. 
 Each merge is implemented with~\cref{lem:rangemerging} in $\cO(\log\log n)$ time.
\end{proof}

As for allowing substitutions of letters that do not occur in $T$, we simply allow blocks of length $1$ that are not substrings of $T$ in block decompositions, corresponding to such letters.
It is readily verified that all the statements above still hold.

Due to~\cref{lem:lcs_three_blocks}, for a maximal block decomposition $(s_1, s_2, \dots, s_k)$ of $S$ with respect to~$T$,
we know that any occurrence of an LCS of $S$ and $T$ in $S$ must contain the first letter of some block of the decomposition and cannot
span more than three blocks.
In other words, it is the concatenation of a potentially empty suffix of $s_{i-1}s_{i}$ and a potentially empty prefix of $s_{i+1}s_{i+2}$
for some $i \in [k]$ (for convenience we consider the non-existent $s_{i}$s to be equal to $\varepsilon$).
We call an LCS that can be decomposed in such way a candidate of $s_i$.
Our goal is to maintain the candidate proposed by each $s_i$ in a max-heap with the length as the key. We also store
a pointer to it from block $s_i$. The max-heap is implemented with an $\cO(n)$-size
predecessor/successor data structure over $[n]$ that supports $\cO(\log \log n)$-time queries and updates~\cite{vanemdeboas1975}.
We assume that each block $s_i$ stores a pointer to its candidate in the max-heap.

After an update, the candidate of each block $b$ that satisfies the following two conditions remains unchanged: (a) $b$
did not change and (b) neither of $b$'s neighbors at distance at most $2$ changed.
For the $\cO(1)$ blocks that changed, we proceed as follows.
First, in $\cO(\log\log n)$ time, we remove from the max-heap any candidates proposed by the deleted blocks or blocks whose
neighbors at distance at most $2$ have changed.
Then, for each new block and for each block whose neighbors at distance at most $2$ have changed, we compute
its candidate and insert it to the max-heap.
To compute the candidate of a block $s_i$, we proceed as follows.
We first compute the longest suffix $U$ of $s_{i-1}s_{i}$ and the longest prefix $V$ of $s_{i+1}s_{i+2}$ that occur in
$T$ in $\cO(\log \log n)$ time using~\cref{lem:rangemerging}.
Then, the problem in scope can be restated as follows: given two fragments $U$ and $V$ of $T$ compute the
longest fragment of $UV$ that occurs in $T$.
This problem can be reduced in $\cO(\log \log n)$ time to a constant number of HIA queries over the suffix trees of $T$ and $T^R$
as shown in~\cite{amir17,DBLP:conf/cpm/AbedinH0T18}.
Here, we briefly explain the reduction to HIA in the interests of self-containment and developing
intuition for the harder problem of maintaining an LCS of two dynamic strings.

\paragraph{Reduction to \textsc{Heaviest Induced Ancestors}.}
Let $\T_1$ and $\T_2$ be the suffix trees of $T \$$ and $T^R \#$, respectively, where $\$$ and $\#$ are sentinel letters not in the alphabet and lexicographically smaller than all other letters.
Note that each suffix of $T \$$ corresponds to a leaf in $\T_1$; similarly for $\T_2$. We label a leaf $v$ of $\T_1$ with the starting position of the suffix of $T\$$ that it represents. For $\T_2$, however, we label the leaf corresponding to $T^R[i \dd ]\#$ with $n-i+2$. Intuitively, if we consider a split
$T = T[ \dd i-1]T[i \dd ]$, the leaves corresponding to $T[i \dd ]\$$ in $\T_1$ and $T[\dd i-1]^R \#$ in $T_2$ get the same label.
Further, let the weight of each node in $\T_1$ and $\T_2$ be its string-depth.
Upon query, we first compute the node $p$ corresponding to $V$ in $\T_1$ and the node $q$ corresponding to $U^R$ in $\T_2$ using weighted
ancestor queries in $\cO(\log\log n)$ time. 
Then the length of the longest substring of $UV$ is exactly the sum of the weights of the nodes returned by a HIA query
for $p$ and $q$. Some technicalities arise when $p$ or $q$ are implicit nodes, which can be overcome straightforwardly.

Let $Q_{\operatorname{HIA}}(n)$, $S_{\operatorname{HIA}}(n)$, and $P_{\operatorname{HIA}}(n)$, respectively, denote upper bounds on the query-time, space, and preprocessing 
of a data structure for the 
\textsc{Heaviest Induced Ancestors} problem when the input is of size $\cO(n)$.
Combining the above discussion with~\cref{lem:blockupd}, we obtain the following result, which is a reduction of the partially dynamic LCS problem to the static HIA problem.

\begin{theorem}\label{thm:onesided}
We can maintain an LCS of a dynamic string $S$ and a static string $T$, each of length at most $n$, in $\cO(\log\log n)+Q_{\operatorname{HIA}}(n)$ time per substitution operation using $\cO(n \log n)+S_{\operatorname{HIA}}(n)$ space, after an $\cO(n\log^2 n)+P_{\operatorname{HIA}}(n)$-time preprocessing.
\end{theorem}

Using~\cref{lem:hia}, we obtain that an LCS can be maintained in $\cO(\frac{\log n}{\log\log n})$ time per update using $\cOtilde(n)$ space,
after an $\cOtilde(n)$ time preprocessing .
Clearly, the bottleneck in the update time
in this approach is in~\cref{lem:hia}. That is, it is the query-time of the data structure for answering HIA queries, as the additional time for each update is only $\cO(\log\log n)$.
We can thus obtain a faster data structure at the expense of slower preprocessing and more space using the following lemma, whose proof is deferred to~\cref{sec:hia2}.

\begin{restatable}{lemma}{lemhia}
\label{lem:hia2}
For any constant $\eps>0$, there is a structure for the \textsc{Heaviest Induced Ancestors} problem,
that can be built in $\cO(n^{1+\eps})$ time and answers queries in $\cO(1)$ time. 
\end{restatable}

By plugging in the data structures encapsulated by~\cref{lem:hia,lem:hia2} to~\cref{thm:onesided} we obtain the following results.

\begin{corollary}\label{cor:onesided}
We can maintain an LCS of a dynamic string $S$ and a static string $T$, each of length at most $n$,
\begin{enumerate}[label=(\alph*)]
\item in $\cO(\log n \log\log n)$ time per substitution operation using $\cO(n \log^2 n)$ space, after an $\cO(n\log^2 n)$-time preprocessing, or
\item in $\cO(\frac{\log n}{\log\log n})$ time per substitution operation using $\cOtilde(n)$ space, after an $\cOtilde(n)$-time preprocessing, or
\item in $\cO(\log\log n)$ time per substitution operation using $\cO(n^{1+\epsilon})$ space, after an $\cO(n^{1+\epsilon})$-time preprocessing, for any constant $\epsilon>0$.
\end{enumerate}
\end{corollary}

\subsection{Handling Insertions and Deletions}
Let us now show how to extend our solution to handle arbitrary edit operations, i.e., insertions and deletions of letters as well.
The sole challenge to this end is in maintaining the maximal block decomposition of $S$.
This is because insertions and deletions offset the indices and hence we cannot use a predecessor structure to identify the block
that contains the edit as in the proof of~\cref{lem:blockupd}.

First, note that we can perform the required task with an additive $\cO(\log n)$ factor overhead in the update time
by maintaining an augmented balanced binary search tree over the blocks.
This way, we appropriately offset their indices upon insertions and deletions of letters while locating the block that contains the edit.
In fact, as we show next, we can perform the required task in $\cO(\log n/ \log\log n)$ amortized time per update using
the dynamic array data structure of Dietz~\cite{DBLP:conf/wads/Dietz89}.\footnote{This problem is very similar to the \textsc{List Representation} problem defined in~\cref{sec:lb_ps}. In fact, the data structure of Dietz~\cite{DBLP:conf/wads/Dietz89} solves the latter problem within the same complexities and this is optimal due to a lower bound by Fredman and Saks~\cite{10.1145/73007.73040}. A reduction from \textsc{List Representation} to \textsc{List Indexing} is only possible when no duplicates are allowed in the ``represented list''.}
Formally, consider the following problem.

\Dynproblem{List Indexing}{A list $L$ of size at most $n$ over universe $[n]$.}{``Delete $y$'' or ``insert $y$ after $x$''.}{``Report the $i$-th element of $L$'' or ``return the position of element $x$ in $L$''.}

\begin{theorem}[\cite{DBLP:conf/wads/Dietz89}]
There is a data structure for \textsc{List Indexing} that, after an $\cO(n)$-time initialization,
can perform updates and answer queries in $\cO(\log n/ \log\log n)$ amortized time.
\end{theorem}

Let us now return to our problem.
In what follows, we will slightly abuse notation by referring to the $i$-th element of a list $L$ as $L[i]$.
We will maintain lists $L_S$, $L_B$, and $L_M$, defined as follows:\footnote{Note that \textsc{List Indexing} assumes that the elements in the list are pairwise distinct. We ensure this by giving pairwise distinct identifiers to the elements of each list.}
\begin{enumerate}
\item $L_S[i]=S[i]$,
\item $L_B[k]$ stores the $k$-th block $s_k$ in the decomposition,
represented as $(a,b)$, where $s_k=T[a \dd b]$,
\item $L_M$ is the concatenation of lists $L_B[k]L_S[i \dd j]$, where $s_k=S[i \dd j]$, in increasing order with respect to~$k$.
\end{enumerate}
We are only interested in $L_M$. We maintain $L_S$ and $L_B$ for technical reasons, for the sake of efficiency.

There is a natural bijection $f$ from the elements of $L_M$ to the elements of the concatenation of $L_S$ and $L_B$.
We will maintain bidirectional pointers according to this bijection.
Formally, between $L_B[k]$ and the delimiter of block $s_k$ in $L_M$,
and between $L_S[i]$ and $L_M[i+b]$, where $b$ is the rank of the block containing $S[i]$.
Such pointers can be implemented using $\cO(1)$ query operations,
provided that we replace each element $x\in \{a,f(a): a\in L_M\}$ by $(x,\textsf{id}(a))$,
where $\textsf{id}(a)$ is a unique identifier.

\begin{example}
Let $S=\texttt{abaabca}$ and $T=\texttt{bcaba}$. Example lists $L_S$, $L_B$, and $L_M$ are shown below with blocks $\texttt{aba}=T[3 \dd 5]$, $\texttt{a}=T[3]$, and $\texttt{bca}=T[1 \dd 3]$. The unique identifiers are highlighted in blue.
\medskip

\noindent $L_S \quad$\scalebox{0.93}{\begin{tabular} {!{\vrule width 1pt}c!{\vrule width 1pt}c!{\vrule width 1pt}c!{\vrule width 1pt}c!{\vrule width 1pt}c!{\vrule width 1pt}c!{\vrule width 1pt}c!{\vrule width 1pt}}
\noalign{\hrule height 1pt}
$(\texttt{a},\textcolor{blue}{5})$ & $(\texttt{b},\textcolor{blue}{1})$ & $(\texttt{a},\textcolor{blue}{7})$ & $(\texttt{a},\textcolor{blue}{10})$ & $(\texttt{b},\textcolor{blue}{2})$ & $(\texttt{c},\textcolor{blue}{8})$ & $(\texttt{a},\textcolor{blue}{4})$ \\ \noalign{\hrule height 1pt}
\end{tabular}}
\medskip

\noindent $L_B \;\;\;$\scalebox{0.93}{\begin{tabular} {!{\vrule width 1pt}c!{\vrule width 1pt}c!{\vrule width 1pt}c!{\vrule width 1pt}}
\noalign{\hrule height 1pt}
{\cellcolor{magenta!35!white}} $((3,5),\textcolor{blue}{6})$ & {\cellcolor{magenta!35!white}} $((3,3),\textcolor{blue}{3})$ & {\cellcolor{magenta!35!white}} $((1,3),\textcolor{blue}{9})$  \\ \noalign{\hrule height 1pt}
\end{tabular}}
\medskip

\noindent $L_M \;\;\,$\scalebox{0.93}{\begin{tabular} {!{\vrule width 1pt}c!{\vrule width 1pt}c!{\vrule width 1pt}c!{\vrule width 1pt}c!{\vrule width 1pt}c!{\vrule width 1pt}c!{\vrule width 1pt}c!{\vrule width 1pt}c!{\vrule width 1pt}c!{\vrule width 1pt}c!{\vrule width 1pt}}
\noalign{\hrule height 1pt}
{\cellcolor{magenta!35!white}} $((3,5),\textcolor{blue}{6})$ & $(\texttt{a},\textcolor{blue}{5})$ & $(\texttt{b},\textcolor{blue}{1})$ & $(\texttt{a},\textcolor{blue}{7})$ & {\cellcolor{magenta!35!white}} $((3,3),\textcolor{blue}{3})$ & $(\texttt{a},\textcolor{blue}{10})$ & {\cellcolor{magenta!35!white}} $((1,3),\textcolor{blue}{9})$ & $(\texttt{b},\textcolor{blue}{2})$ & $(\texttt{c},\textcolor{blue}{8})$ & $(\texttt{a},\textcolor{blue}{4})$ \\ \noalign{\hrule height 1pt}
\end{tabular}}

\end{example}

Let us now discuss how to perform a deletion in $S$. Substitutions and insertions can be handled analogously.
Upon the deletion of a letter $S[i]=a$, we follow the pointer from $L_S[i]$ to some element $L_M[j]$ of $L_M$.
We delete both $L_S[i]$ and $L_M[j]$.
It follows that $S[i]$ is on the $k$-th block in the decomposition, where $k=j-i$.
This block $s_k$ is stored in $L_B[k]$.
We split this block to at most two blocks $s_k^\textrm{l}$ and $s_k^\textrm{r}$ such that $s_k=s_k^\textrm{l} a s_k^\textrm{r}$ similarly to the proof of~\cref{lem:blockupd}, and then repeatedly try to merge consecutive blocks in $(s_{k-1},s_k^\textrm{l},s_k^\textrm{r},s_{k+1})$.

When we cannot merge any more blocks, we update lists $L_B$ and $L_M$ as follows. 
We obtain the set $D$ of delimiters of blocks $s_{k-1}$, $s_k$ and $s_{k+1}$ in $L_M$ using the pointers from $L_B[k-1]$, $L_B[k]$, and $L_B[k+1]$, respectively.
We delete any of these delimiters that should be deleted.
The indices where delimiters are to be inserted in $L_M$ (corresponding to new blocks) can be straightforwardly obtained from $D$ and the lengths of the new blocks.
In the end, we create bidirectional pointers between new blocks and their corresponding delimiters,
getting identifiers from a global max-heap of available identifiers.
We obtain the following result.

\begin{theorem}\label{thm:onesided_edit}
We can maintain an LCS of a dynamic string $S$ and a static string $T$, each of length at most $n$, in $\cO(\log n)+Q_{\operatorname{HIA}}(n)$ time or $\cO(\log n/\log \log n)+Q_{\operatorname{HIA}}(n)$ amortized time per edit operation using $\cO(n \log n)+S_{\operatorname{HIA}}(n)$ space, after an $\cO(n\log^2 n)+P_{\operatorname{HIA}}(n)$-time preprocessing.
\end{theorem}

We obtain the following results by plugging in the data structures encapsulated by~\cite{DBLP:conf/cccg/GagieGN13} or \cref{lem:hia} to~\cref{thm:onesided_edit}.

\begin{corollary}\label{cor:onesided_edit}
We can maintain an LCS of a dynamic string $S$ and a static string $T$, each of length at most $n$,
\begin{enumerate}[label=(\alph*)]
\item in $\cO(\log n \log\log n)$ time per edit operation using $\cO(n \log^2 n)$ space, after an $\cO(n\log^2 n)$-time preprocessing, or
\item in $\cO(\log n/ \log\log n)$ amortized time per edit operation using $\cOtilde(n)$ space, after an $\cOtilde(n)$-time preprocessing, for any constant $\epsilon>0$.
\end{enumerate}
\end{corollary}

\subsection{\textsc{Heaviest Induced Ancestors} in Constant Time (Proof of~\cref{lem:hia2})}\label{sec:hia2}

\lemhia*
\begin{proof}
Consider an instance of HIA on two trees $\T_1$ and $\T_2$ of total size $m$ containing
at most $\ell$ leaves, and let $b$ be a parameter to be chosen later. We will show how to
construct a structure of size $\cO(b^{2}m)$ that allows us to reduce in constant time a query
concerning two nodes $u \in \T_1$ and $v \in \T_2$ to two queries to instances of HIA on smaller trees.
In each of the smaller instances the number of leaves will shrink by a factor of $b$, and the
total size of all smaller instances will be $\cO(m)$. Let $b=n^{\delta}$, where $n$ is the total
size of the original trees. We recursively repeat the construction, always choosing $b$ according to the formula.
Because the depth of the recursion is at most $\log_{b}n=\cO(1)$, this results in a structure of total size
$\cO(n^{1+\eps})$ for $\eps=2\delta$ and allows us to answer any query in a constant number of steps, each
taking constant time.

We select $b$ evenly-spaced (in the order of in-order traversal) leaves of $\T_{1}$ and $\T_{2}$ and call them marked.
Consider a query concerning a pair of nodes $u\in\T_{1}$ and $v\in \T_{2}$. Let $u'',v''$
be the nearest ancestors of $u$ and $v$, respectively, that contain at least one marked leaf
in their subtrees.
$u''$ and $v''$ can be retrieved in constant time after an $\cO(m)$-time preprocessing.
We have three possibilities concerning the sought ancestors $u'$ and $v'$:
\begin{enumerate}
\item $u'$ is an ancestor of $u''$ and $v'$ is an ancestor of $v''$ (not necessarily strict),
\item $u'$ is a descendant of $u''$,
\item $v'$ is a descendant of $v''$.
\end{enumerate}

Before showing how to treat the first case, let us state the following claim, which was
(implicitly) shown in the work that introduced the HIA problem~\cite{DBLP:conf/cccg/GagieGN13}.

\begin{claim}[\cite{DBLP:conf/cccg/GagieGN13}, Section 2.1]
Given a pair of paths $p_1$ in $\T_{1}$ and $p_2$ in $\T_{2}$, we can construct in $\cO(m)$ time a data structure, and two sets $A_1 \subseteq \{w(y) : y \in p_1\}$ and $A_2 \subseteq \{w(y) : y \in p_2\}$,
so that a HIA query for a node $z_1$ in $p_1$ and a node $z_2$ in $p_2$ can be answered in $\cO(1)$ time, plus the time required to find the predecessor of $w(z_1)$ in $A_1$
and the predecessor of $w(z_2)$ in $A_2$.
\end{claim}

For every pair of marked leaves $x$ in $\T_{1}$ and $y$ in $\T_{2}$ we do the following.
Let $p_1$ be the $\text{root-to-}x$ path and $p_2$ be the $\text{root-to-}y$ path.
We employ the above claim for the pair of paths $(p_1,p_2)$, and also precompute the predecessor of each element of $p_1$ in $A_1$ and each element of $p_2$ in $A_2$.
This can be done in $\cO(m)$ time.
Over all pairs of marked leaves, the total preprocessing time and space are $\cO(b^{2}m)$.

For the first case, after retrieving a marked leaf in the subtree of each of $u''$ and $v''$, we employ the data structure we described above for that pair of marked leaves, hence answering the query in $\cO(1)$ time. Note that, for each node we can precompute a pointer to a marked leaf in its subtree in $\cO(m)$ total time.

The second and the third cases are symmetric, so let us focus on the second one. By removing
all marked leaves and their ancestors from $\T_{1}$ we obtain a collection of smaller
trees, each containing less than $n/b$ leaves. Because $u'$ is below $u''$, $u$ and $u'$
belong to the same smaller tree. For technical reasons, we want to work with $\cO(b)$
smaller trees, so we merge all smaller trees between two consecutive marked
leaves by adding the subtree induced by their roots in~$\T_{1}$. Now, let us consider the smaller
tree $\T_{1}^{i}$ containing $u$ (and, by assumption, also $u''$). We extract the subtree
of $\T_{2}$ induced by the leaves of $\T_{1}^{i}$, call it $\T_{2}^{i}$, and build a smaller
instance of HIA for $\T_{1}^{i}$ and $\T_{2}^{i}$. To query the smaller instance, we need
to replace $v$ by its nearest ancestor that belongs to $\T_{2}^{i}$. This can be preprocessed
for each $i$ and $v$ in $\cO(bm)$ time. By construction, $\T_{1}^{i}$ and $\T_{2}^{i}$ contain
less than $n/b$ leaves, and each node of $\T_{1}$ shows up in at most two trees $\T_{1}^{i}$.
Each node of $\T_{2}$ might appear in multiple trees $\T_{2}^{i}$, but the number of non-leaf
nodes in $\T_{2}^{i}$ is smaller than its number of leaves, so the overall number of non-leaf
nodes is smaller than $m$, and consequently the overall number of nodes is smaller than $2m$.
\end{proof}

\section{Fully Dynamic LCS}
\label{sec:fully}

In this section, we prove our main result.

\begin{theorem}\label{thm:fully}
We can maintain an LCS of two initially empty dynamic strings, each
of length at most $n$, in $\cO(\log^7 n)$ amortized time per edit operation with high probability, using $\cOtilde(n)$ space.\footnote{Our algorithm is Las Vegas randomized, that is, it always returns correct answers but its running time holds with probability $1-1/n^{\Omega(1)}$.}
\end{theorem}

We start with some intuition. Let us suppose that we can maintain a decomposition of each string in level-$k$ blocks of length roughly $2^k$ for each level $k=0,1, \ldots, \log n$ with the following property: any two equal fragments $U=S[i \dd j]$ and $V=T[i' \dd j']$ are ``aligned'' by a pair of equal blocks $B_1$ in $S$ and $B_2$ in $T$ at some level $k$ such that $2^k=\Theta(|U|)$.
In other words, the decomposition of $U$ (resp.~$V$) at level $k$ consists of a constant number of blocks, where the first and last blocks are potentially trimmed, including $B_1$ (resp.~$B_2$), and the distance of the starting position of $B_1$ from position $i$ in $S$ equals the distance of the starting position of $B_2$ from position $i'$ in $T$.
The idea is that we can use such blocks as anchors for the LCS. For each level, for each string $B$ appearing as a block in this level, we would like to design a data structure that:
\begin{enumerate}
\item supports insertions/deletions of strings corresponding to sequences of a constant number of level-$k$ blocks, each containing a specified block equal to $B$ and a Boolean variable indicating the string this sequence originates from ($S$ or $T$), and
\item can return the longest common substring among pairs of elements originating from different strings that is aligned by a pair of blocks equal to $B$.
\end{enumerate}
For each edit operation in either of the strings, we would only need to update $\cO(\log n)$ entries in our data structures -- a constant number of them per level.

Unfortunately, it is not clear how to maintain a decomposition with these properties.
We resort to the dynamic maintenance of a \emph{locally consistent parsing} of the two strings, due to Gawrychowski et al.~\cite{ods}.
We exploit the structure of this parsing in order to apply the high-level idea outlined above in a much more technically demanding setting.

\subsection{Locally Consistent Parsing}

The authors of~\cite{ods} settled the time complexity of maintaining a collection of strings $\mathcal{W}$ under the following operations:
\textsf{makestring}$(W)$ (insert a non-empty string $W$),
\textsf{concat}$(W_1,W_2)$ (insert $W_1 W_2$ to $\mathcal{W}$, for $W_1,W_2 \in \mathcal{W}$),
\textsf{split}$(W,i)$ (split the string $W$ at position $i$ and insert both resulting strings to $\mathcal{W}$, for $W \in \mathcal{W}$),
\lcp$(W_1,W_2)$ (return the length of the longest common prefix of $W_1$ and $W_2$, for $W_1,W_2 \in \mathcal{W}$).
Let us note that the collection is persistent in the sense that
operations \textsf{concat} and \textsf{split} do not remove their arguments from $\mathcal{W}$.
An edit operation can be implemented with a constant number of calls to such operations.

\begin{theorem}[Gawrychowski et al.~\cite{ods}]\label{thm:dynst}
A persistent dynamic collection $\mathcal{W}$ of strings of total length $n$ can be maintained under operations \textsf{makestring}$(W)$, \textsf{concat}$(W_1,W_2)$, \textsf{split}$(W,i)$, and \lcp$(W_1,W_2)$ with the operations requiring $\cO(\log n +|W|)$, $\cO(\log n)$, $\cO(\log n)$
worst-case time with high probability and $\cO(1)$ worst-case time, respectively.
\end{theorem}

The data structure underlying the above theorem is Las Vegas randomized:
the answers are correct, but the update times are guaranteed only with probability $1-1/m^{\Omega(1)}$, where $m$ is the total size of the input, including updates.
Note that the size of a $\makestring(W)$ operation is $|W|$, while the size of a \textsf{concat} or a \textsf{split} operation is $\cO(1)$.
We will use~\cref{thm:dynst} to maintain two dynamic strings of length at most $n$ under edit operations: each update to the collection will increase the total size of the collection by an $\cO(n)$ additive factor. By rebuilding the data structure every $\cO(n)$ updates, we can ensure that the update times are guaranteed with probability  $1-1/n^{\Omega(1)}$.

At the heart of~\cref{thm:dynst} lies a locally consistent parsing of the strings in the collection that can be maintained efficiently while the strings undergo updates.
It can be interpreted as a dynamic version of the recompression method of Je{\.{z}}~\cite{DBLP:journals/jacm/Jez16,DBLP:journals/talg/Jez15} (see also~\cite{DBLP:conf/cpm/I17}) for a static string~$T$.
As such, we first describe the parsing of~\cref{thm:dynst} for a static string $T$ and then extend the description to the dynamic variant for a collection of strings.

A \emph{run-length straight line program} (\emph{RLSLP}) is a context-free grammar which generates exactly one string
and contains two kinds of non-terminals: \emph{concatenations} with production rule of the form $A \to BC$ (for distinct symbols $B,C$)
and \emph{powers} with production rule of the form $A \to B^k$ (for a symbol $B$ and an integer $k\ge 2$), where a \emph{symbol} can be a non-terminal or a letter in $\Sigma$.
Every symbol $A$ generates a unique string denoted by $\gen(A)$.
We call a symbol $A$ that is a non-terminal a \emph{concatenation symbol} or a \emph{power symbol} depending on its production.

Let $T=T_0$. We can compute strings $T_1, \ldots, T_H$, where $H=\cO(\log n)$ and $|T_H|=1$ in $\cO(n)$ time using interleaved calls to the following two auxiliary procedures:
\begin{description}
    \item[\textsf{RunCompress}] applied if $h$ is even: for each $B^r$, $r>1$, replace all occurrences of $B^r$ as a run by a new letter $A$. There are no runs after an application of this procedure.\footnote{A fragment $T[i \dd j]=B^r$ is a run if it is a maximal fragment consisting of $B$s.}
    \item[\textsf{HalfCompress}] applied if $h$ is odd: first partition $\Sigma$ into $\Sigma_{\ell}$ and $\Sigma_r$; then, for each pair of letters $B \in \Sigma_{\ell}$ and $C \in \Sigma_r$ such that $BC$ occurs in $T_h$ replace all occurrences of $BC$ by a new letter $A$.
\end{description}

We can interpret strings $T=T_0, T_1, \ldots, T_H$ as an \emph{uncompressed parse tree} $\PT(T)$, by considering their letters as nodes, so that the parent of $T_h[i]$ is the letter of $T_{h+1}$ that either (a) corresponds to $T_h[i]$ or (b) replaced a fragment of $T_h$ containing $T_h[i]$.
We say that the node representing $T_h[i]$ is the node left (resp.~right) of the node representing $T_h[i+1]$ (resp.~$T_h[i-1]$).
Every node $v$ of $\PT(T)$ is labeled with the symbol it represents, denoted by $\Lv(v)$.
For a node~$v$ corresponding to a letter of $T_h$, we say that the level of $v$, denoted by $\lev(v)$, is $h$.
The \emph{value} $\val(v)$ of a node $v$ is defined as the fragment of $T$ corresponding to the leaf descendants of $v$ and it is an occurrence of $\gen(A)$ for $A=\Lv(v)$.

We define a \emph{layer} to be any sequence of nodes $v_1 v_2 \cdots v_r$ in $\PT(T)$ whose values are consecutive fragments of $T$, i.e., $\val(v_j)=T[r_{j-1}+1 \dd r_j]$ for some increasing sequence of $r_i$'s.
The value of a layer $C$ is the fragment obtained by concatenating the values of $C$'s elements and is denoted by $\val(C)$.
We similarly use $\gen(\cdot)$ for sequences of symbols, to denote the concatenation of the strings generated by them.
We call a layer $ v_1 v_2 \cdots  v_r$ an \emph{up-layer} when $\lev(v_i)\leq \lev(v_{i+1})$ for all $i$, and a \emph{down-layer} when $\lev(v_i)\geq \lev(v_{i+1})$ for all $i$.

In~\cite{ods}, the authors show how to maintain an RLSLP for each string in the collection, each with at most $c \log n$ levels for some global constant $c$;
when this condition is no longer satisfied, the data structure is reinitialized.
Let $T$ be a string in the collection.
For each fragment $U=T[a \dd b]$ of $T$, one can compute in $\cO(\log n)$ time a layer $\C(U)$ of nodes in $\PT(T)$ that has value $T[a \dd b]$ and satisfies the following property.
It can be decomposed into an up-layer $\CUp(U)$ and a down-layer $\CDown(U)$ such that:
\begin{itemize}
\item The sequence of the labels of the nodes in $\CUp(U)$ can be expressed as a sequence of at most $c \log n$ symbols and powers of symbols $\dUp(U)=A_0^{r_0} A_1^{r_1} \cdots  A_m^{r_m}$ such that, for all $i$, $A_i^{r_i}$ corresponds to $r_i$ consecutive nodes at level $i$ of $\PT(T)$; $r_i$ can be $0$ for $i< m$.
\item Similarly, the sequence of the labels of the nodes in $\CDown(U)$ can be expressed as a sequence of at most $c \log n$ symbols and powers of symbols $\dDown(U)=B_m^{t_m} B_{m-1}^{t_{m-1}} \cdots B_0^{t_0}$ such that, for all $i$, $B_i^{t_i}$ corresponds to $t_i$ consecutive nodes at level $i$ of $\PT(T)$; $t_i$ can be equal to $0$.
\end{itemize}
We denote by $\dUpDown(U)$ the concatenation of $\dUp(U)$ and $\dDown(U)$.
Note that $U=\gen(\dUpDown(U))=\gen(A_0)^{r_0} \cdots \gen(A_m)^{r_m} \gen(B_m)^{t_m} \cdots \gen(B_0)^{t_0}$.  See~\cref{fig:up-down-layer} for a visualization.
The parsing of the strings enjoys local consistency in the following way: $\dUpDown(U)=\dUpDown(V)$ for any fragment $V$ of any string in the collection such that $U=V$.
We will use $\dUpDown(\cdot)$ for substrings and not just for fragments.

\FIGURE{ht}{0.7}{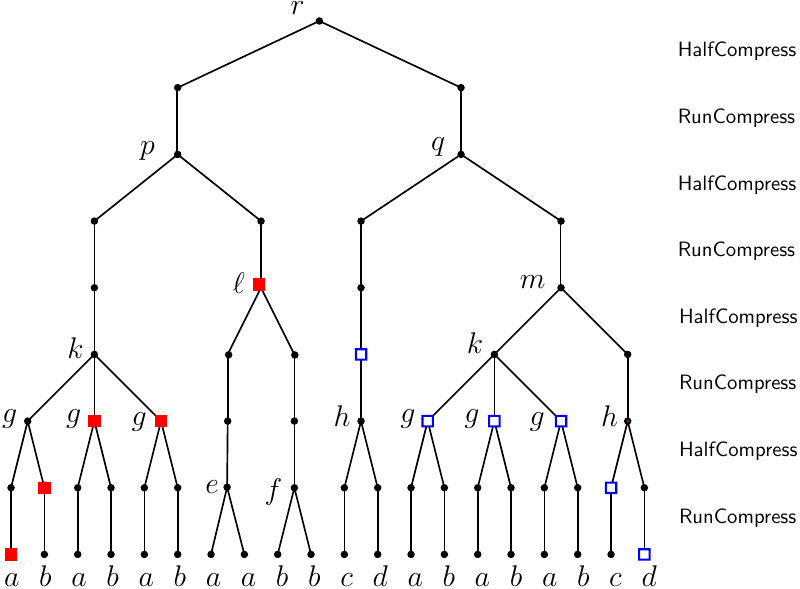}{
An example $\PT(T)$ for $T=T_0=abababaabbcdabababcd$. We omit the label of each node $v$ with a single child $u$; $\Lv(v)=\Lv(u)$. $T_3=kefhkh$ and $T_6=pq$. We denote the nodes $\CUp(T)$ by red (filled) squares and the nodes of $\CDown(T)$ with blue (unfilled) squares. $\dUp(T)=a b g^2 \ell$, $\dDown(T)=hg^3cd$ and hence $\dUpDown(T)=ab g^2 \ell h g^3 cd$.}

Let us consider any sequence of nodes corresponding, for some $j<m$, to $A_j^{r_j}$ with $r_j>1$ or $B_j^{t_j}$ with $t_j>1$. We note that $T_j$ must have been obtained from $T_{j-1}$ by an application of \textsf{HalfCompress}, since there are no runs after an application of procedure \textsf{RunCompress}.
Thus, at level $j+1$ in $\PT(T)$, i.e., the one corresponding to $T_{j+1}$, all of these nodes collapse to a single one: their parent in $\PT(T)$.
Hence, we have the following lemma.

\begin{lemma}\label{lem:layers}
Let us consider a fragment $U$ of $T$ with $\dUp(U)=A_0^{r_0} A_1^{r_1} \cdots A_m^{r_m}$ and $\dDown(U)=B_m^{t_m} B_{m-1}^{t_{m-1}} \cdots B_0^{t_0}$. Then we have the following:
\begin{itemize}
\item The value of $\CUp(U)$ is a suffix of the value of a layer of (at most) $c \log n + r_m-1$ level-$m$ nodes, such that the two layers have the same rightmost node. The last $r_m$ nodes of each of these two layers are consecutive siblings with label $A_m$.
\item The value of $\CDown(U)$ is a prefix of the value of the layer consisting of the subsequent (at most) $c \log n+\max(t_m-1,0)$ level-$m$ nodes. If $t_m \neq 0$, then the first $t_m$ nodes of both of these layers are consecutive siblings with label $B_m \neq A_m$.
\end{itemize}
\end{lemma}

The parse trees of the strings in the collection are not maintained explicitly.
However, we have access to the following pointers and functions, among others, which allow us to efficiently navigate through them. First, we can get a pointer to the root of $\PT(T)$ for any string $T$ in the collection. Given a pointer $P$ to some node $v$ in $\PT(T)$ we can get $\deg(v)$ and pointers to the parent of $v$, the $k$-th child of $v$ and the nodes to the left/right of $v$.

Let us now briefly explain how the dynamic data structure of~\cite{ods} processes a substitution in~$T$ at some position $i$, that yields a string $T'$.
First, $\C(T[\dd i-1])$ and $\C(T[i+1 \dd ])$ are retrieved.
These, together with the new letter at position $i$ form a layer of $\PT(T')$.
The sequence of the labels of the nodes of this layer can be expressed as a sequence of $\cO(\log n)$ symbols and powers of symbols.
Then, only the portion of $\PT(T)$ that lies above this layer needs to be (implicitly) computed, and the authors of~\cite{ods} show how to do this in $\cO(\log n)$ time.
In total, we get $\PT(T')$ from $\PT(T)$ by performing the following update in each level $h$ of the parse tree:
a fragment of $T_h$ of the form $A^a U B^b$, where $A$ and $B$ are letters of $T_h$ and $U$ is a fragment of length $\cO(\log n)$
is replaced by a fragment of $T'_h$ of the same form.
In the parse tree view, over all levels, we have $\cO(\log^2 n)$ insertions and deletions of nodes and layers that consist of consecutive siblings.

Note that the described substitution process can be seen as the processing of a deletion followed by a processing of an insertion --- insertions and deletions of letters in $T$ can also be processed analogously in $\cO(\log n)$ time.
These also result, for each level $h$, in the replacement of a fragment of $T_h$ of the form $A^a U B^b$, where $A$ and $B$ are letters of $T_h$ and $U$ is a fragment of length $\cO(\log n)$
with a fragment of $T'_h$ of the same form.

\subsection{Anchoring the LCS}

We will rely on~\cref{lem:layers} in order to identify an LCS $S[i \dd j]=T[i' \dd j']$ at a pair of topmost nodes of $\C(S[i \dd j])$ and $\C(T[i' \dd j'])$ in $\PT(S)$ and $\PT(T)$, respectively.
In order to develop some intuition, let us first sketch a solution for the case where $\PT(S)$ and $\PT(T)$ do not contain any power symbols throughout the execution of our algorithm.
For each node $v$ in one of the parse trees, let $\ZL(v)$ be the value of the layer consisting of the (at most) $c \log n$ level-$\lev(v)$ nodes, with $v$ being the layer's rightmost node, and
$\ZR(v)$ be the value of the layer consisting of the (at most) $c \log n$ subsequent level-$\lev(v)$ nodes.
Now, consider a common substring $X$ of $S$ and $T$ and partition it into the prefix $X_\ell=\gen(\dUp(X))$ and the suffix $X_r=\gen(\dDown(X))$.
For any fragment $U$ of $S$ that equals $X$, $\CUp(U)$ is an up-layer of the form $v_1 \cdots v_m$.
Hence, by~\cref{lem:layers}, $X_\ell$ is a suffix of $\ZL(v_m)$. Similarly, $X_r$ is a prefix of $\ZR(v_m)$.
Thus, it suffices to maintain pairs $(\ZL(v),\ZR(v))$ for all nodes $v$ in $\PT(S)$ and $\PT(T)$, and, in particular, a pair of nodes $u \in \PT(S)$ and $v \in \PT(T)$ that 
maximizes $\lcp(\ZL(u)^R,\ZL(v)^R)+\lcp(\ZR(u),\ZR(v))$.
The existence of power symbols poses some technical challenges which we overcome below.

For each node of $\PT(T)$, we consider at most one pair consisting of an up-layer and a down-layer.
Crucially, for each maximal set of consecutive siblings, we consider a number of pairs proportional to the minimum of their number and $\log n$;
considering one pair for each of the siblings would be correct but inefficient.
This leads to a different treatment of nodes based on their parent.
We have two cases.

\begin{enumerate}
\item For each node $z$ with $\deg(z)=2$ such that $\Lv(z)$ is a concatenation symbol, we consider the following layers:
\begin{itemize}
\item The layer $\JUp(v)$ of the (at most) $c\log n$ level-$\lev(v)$ consecutive nodes of $\PT(T)$ with $v$ a rightmost node.
\item The layer $\JDown(v)$ of the (at most) $c\log n + \deg(w)$ subsequent level-$\lev(v)$ nodes of $\PT(T)$, where $w$ is the parent of the node to the right of $v$.
\end{itemize}

\item For each node $z$ of $\PT(T)$ whose label is a power symbol and has more than one child, we will consider $\cO(\log n)$ pairs of layers.
In particular, for each $v$, being one of the $c \log n+1$ leftmost or $c \log n+1$ rightmost children of $z$, we consider the following layers:
\begin{itemize}
\item The layer $\JUp(v)$ consisting of (a) the (at most) $c\log n$ level-$\lev(v)$ consecutive nodes of $\PT(T)$ preceding the leftmost child of $z$ and (b) all the children of $z$ that lie weakly to the left of $v$, that is, including $v$.
\item The layer $\JDown(v)$ consisting of the (at most) $c\log n$ subsequent level-$\lev(v)$ nodes of $\PT(T)$ -- with one exception.
If $v$ is the rightmost child of $z$ and the node to its right is a child of a node $w$ with more than two children,
then $\JDown(v)$ consists of the $c\log n+\deg(w)$ subsequent level-$\lev(v)$ nodes.
\end{itemize}
\end{enumerate}

In particular, we create at most one pair $(\JUp(v), \JDown(v))$ of layers for each node $v$ of $\PT(T)$.
Let $\YL(v)=\val(\JUp(v))$ and $\YR(v)=\val(\JDown(v))$.
Given a pointer to a node $z$ in $\PT(T)$, we can compute the indices of the fragments corresponding to those layers with a straightforward use of the pointers at hand in $\cO(\log n)$ time. 
With a constant number of split operations, we can then add the string $\YR(v)$ to our collection within $\cO(\log n)$ time.
Similarly, if we also maintain~$T^R$ in our collection of strings, we can add the reverse of $\YL(v)$ to the collection within $\cO(\log n)$ time.
We maintain pointers between $v$ and these strings.
Note that each node of $\PT(T)$ takes part in $\cO(\log n)$ pairs of layers and these pairs can be retrieved in $\cO(\log n)$ time.
Similarly, for each node whose label is a power symbol, subsets of its children appear in $\cO(\log n)$ pairs of layers; these can also be retrieved in $\cO(\log n)$ time.
These pairs of layers (or rather the pairs of their corresponding strings maintained in a dynamic collection) will be stored in an abstract structure presented in the next section.

Recall that each update on $T$ is processed in $\cO(\log n)$ time.
In each level of the parse tree, it deletes and then inserts $\cO(\log n)$ nodes and, possibly, two additional layers of consecutive siblings,
replacing a fragment of $T_h$.
In each level, the total number of affected pairs of layers is thus $\cO(\log n)$, for a total of $\cO(\log^2 n)$ pairs which can be computed in $\cO(\log^2 n)$ time.
The total time required to add/remove the affected pairs of layers is thus $\cO(\log^3 n)$.
In order to keep the space occupied by our data structure $\cOtilde(n)$, we split the updates into phases: after every $n$ updates to the collection, we delete our data structure and initialize a new instance of it for an empty collection, on which we call \textsf{makestring}$(S)$ and \textsf{makestring}$(T)$ and begin a new phase.
The cost of this reinitialization can be deamortized using standard techniques, by maintaining two copies of the structure, one of them being used to answer queries and the other one that is prepared in the background at appropriate rate, so it is ready when the next phase starts.
We summarize the above discussion in the following lemma.

\begin{lemma}\label{lem:Jup}
We can maintain pairs $(\YL(v)^R, \YR(v))$ for all $v$ in $\PT(T)$ and $\PT(S)$, with each string given as a handle from the dynamic collection,
in $\cO(\log^3 n)$ time with high probability per edit operation, using $\cOtilde(n)$ space. The number of deleted and inserted pairs for each edit operation is $\cO(\log^2 n)$.
\end{lemma}

The following lemma gives us an anchoring property, which is crucial for our approach.

\begin{lemma}\label{lem:anchored}
For any common substring $X$ of $S$ and $T$, there exists a partition $X=X_{\ell}X_r$ for which there exist nodes $u \in \PT(S)$ and $v \in \PT(T)$ such that:
\begin{enumerate}
\item $X_{\ell}$ is a suffix of $\YL(u)$ and $\YL(v)$, and
\item $X_{r}$ is a prefix of $\YR(u)$ and $\YR(v)$.
\end{enumerate}
\end{lemma}
\begin{proof}
Let $\dUp(X)=A_0^{r_0} A_1^{r_1} \cdots A_m^{r_m}$ and $\dDown(X)=B_m^{t_m} B_{m-1}^{t_{m-1}} \cdots B_0^{t_0}$.

\begin{claim}
Either $r_m>1$, $t_m=0$, and $\gen(\dUp(X))$ is not a suffix of $A_m^{c\log n+r_m}$, or there exists a node $v \in \PT(T)$ such that:
\begin{enumerate}
\item $\gen(\dUp(X))$ is a suffix of $\YL(v)$, and
\item $\gen(\dDown(X))$ is a prefix of $\YR(v)$.
\end{enumerate}
\end{claim}
\begin{proof}
We assume that $r_m=1$ or $\gen(\dUp(X))$ is a suffix of $A_m^{c\log n+r_m}$ or $t_m \neq 0$ and distinguish between the following cases.

\emph{Case 1.}
There exists an occurrence $Y$ of $X$ in $T$, where the label of the parent of the rightmost node $u$ of $\CUp(Y)$ is not a power symbol.
(In this case $r_m=1$.)
Recall here, that we did not construct any pairs of layers for nodes whose parent has a single child.
Let $v$ be the lowest ancestor of $u$ with label $A_m$.
If $u \neq v$ then all nodes that are descendants of $v$ and strict ancestors of $u$ have a single child, while the parent of $v$ does not.
In addition, the label of the parent of $v$ must be a concatenation symbol, since only new letters are introduced at each level and thus we cannot have new nodes with label $A_m$ appearing to the left/right of any strict ancestor of $u$.
Finally, note that a layer of $k$ level-$\lev(v)$ nodes with $v$ a leftmost (resp.~rightmost) node contains an ancestor of each of the nodes in a layer of $k$ level-$\lev(u)$ nodes with $u$ a leftmost (resp.~rightmost) node.
Thus, an application of~\cref{lem:layers} for $u$ straightforwardly implies our claim for $v$.

\emph{Case 2.}
There exists an occurrence $Y$ of $X$ in $T$, where the label of the parent $z$ of the rightmost node $u$ of $\CUp(Y)$ is a power symbol.
Let $W$ be the rightmost occurrence of $X$ in $T$ with the rightmost node $w$ of $\CUp(W)$ being a child of $z$.
We have three subcases.
\begin{enumerate}[a)]
\item We first consider the case $r_m=1$.
Let us assume towards a contradiction that $u$ is not one of the $c \log n+1$ leftmost or the $c \log n+1$ rightmost children of $z$.
Then, by~\cref{lem:layers} we have that $\gen(\dUp(X))$ is a suffix of $A_m^{c\log n}$ and $\gen(\dDown(X))$ is a prefix of $A_m^{c\log n}$.
Hence, there is another occurrence of $X$ $|\gen(A_m)|$ positions to the right of $Y$, contradicting our assumption that $Y$ is a rightmost occurrence.
\item In the case where $t_m \neq 0$, $u$ must be the rightmost child of $z$ since $A_m \neq B_m$.
\item In the remaining case where $\gen(\dUp(X))$ is a suffix of $A_m^{c\log n+r_m}$,
either $t_m>0$ and we are done, or $\gen(\CDown(Y))$ is a prefix of the value of the (at most) $c \log n$ level-$m$ nodes to the right of $u$. In the latter case, either $u$ is already among the rightmost $c \log n+1$ children of $z$ or there is another occurrence of $X$ $|\gen(A_m)|$ positions to the right of $Y$, contradicting our assumptions on $Y$.
Here, again, $v$ from the claim is the subsequent level-$\lev(u)$ node from $u$.
\qedhere
\end{enumerate}
\end{proof}

We have to treat a final case.

\begin{claim}
If $r_m>1$, $t_m=0$, and $\gen(\dUp(X))$ is not a suffix of $A_m^{c\log n+r_m}$ then there exists a node $v \in \PT(T)$ such that:
\begin{enumerate}
\item $\gen(A_0^{r_0} A_1^{r_1} \cdots A_{m-1}^{r_{m-1}}A_m)$ is a suffix of $\YL(v)$, and
\item $\gen(A_m^{r_m-1})\gen(\dDown(X))$ is a prefix of $\YR(v)$.
\end{enumerate}
\end{claim}
\begin{proof}
In any occurrence $Y$ of $X$ in $T$, the label of the parent $z$ of the rightmost node of $\CUp(Y)$ is a power symbol. Let $u$ be the $r_m$-th rightmost node of $\CUp(Y)$. By the assumption that $\gen(\dUp(X))$ is not a suffix of $A_m^{c\log n+r_m}$ and~\cref{lem:layers}, $u$ must be one of the $c \log n$ leftmost children of $z$
and $v$ being the subsequent level-$\lev(u)$ node from $u$ satisfies the claim.
\end{proof}

The combination of the two claims applied to both $S$ and $T$ yields the lemma.
\end{proof}

\subsection{A Problem on Dynamic Bicolored Trees}

Due to~\cref{lem:Jup,lem:anchored}, our task reduces to solving the problem defined below in polylogarithmic time per update,
as we can directly apply it to $\mathcal{R}=\{(\YL(u)^{R},\YR(u)): u \in \PT(S)\}$ and $\mathcal{B}=\{(\YL(v)^{R},\YR(v)): v \in \PT(T)\}$.
Note that $|\mathcal{R}|+|\mathcal{B}|=\cOtilde(n)$ throughout the execution of our algorithm.

\Dynproblem{LCP for Two Families of Pairs of Strings}{Two families $\mathcal{R}$ and $\mathcal{B}$, each
consisting of pairs of strings, where each string is given as a handle from a dynamic collection.}{Insertion or deletion of an element
in $\mathcal{R}$ or $\mathcal{B}$.}{Return $(P,Q) \in \mathcal{R}$ and $(P',Q') \in \mathcal{B}$ that maximize $\lcp(P,P')+\lcp(Q,Q')$.}

Each element of $\mathcal{B}$ and $\mathcal{R}$ is given a unique identifier.
We maintain two compacted tries $\T_{P}$ and $\T_{Q}$.
By appending unique letters, we can assume that no string is a prefix of another string.
$\T_{P}$ (resp.~$\T_Q$) stores the string $P$ (resp.~$Q$) for every $(P,Q)\in \mathcal{R}$, with the corresponding leaf colored red
and labeled by the identifier of the pair and the string $P'$ (resp.~$Q'$) for every $(P',Q')\in \mathcal{B}$, with the corresponding leaf colored blue and labeled by the identifier of the pair.
Then, the sought result corresponds to a pair of nodes $u\in \T_{P}$ and $v\in \T_{Q}$ returned by a query to a data structure for the \textsc{Dynamic Bicolored Trees Problem} defined below for $\T_1=\T_P$ and $\T_2=\T_Q$,
with node weights being their string-depths.

\Dynproblem{Dynamic Bicolored Trees Problem}{Two weighted trees $\T_1$ and $\T_2$ of total size at most $m$, whose leaves are bicolored and labeled, so that each label corresponds to exactly one leaf of each tree.}{Split an edge into two / attach a new leaf to a node / delete a leaf.}{Return a pair of nodes $u\in \T_1$ and $v\in \T_2$ with the maximum combined weight that have at least one red descendant with the same label, and at least one blue descendant with the same label.}

Here the input for the split operation is the~leaf $\ell$ and the string depth
$d$.
One can easily locate the actual edge to~split using a weighted ancestor
query from $\ell$.
After such split, the edge $(u, v)$ is decomposed into two edges $(u, w)$ and
$(w, v)$, such that the string depth of $w$ is $d$.

\begin{remark}
A static version of this problem has been used for approximate LCS under the Hamming distance, e.g.~in~\cite{DBLP:conf/cpm/Charalampopoulos18}.
\end{remark}

To complete the reduction, we have to show how to translate an update in $\mathcal{R}$ or $\mathcal{B}$ into updates in $\T_P$ and $\T_Q$.
Let us first explain how to represent $\T_{P}$ and $\T_{Q}$. For each edge, we store
a handle to a string from the dynamic collection, and indices for a fragment of this string which represents the edge's label. 
For each explicit node, we store edges
leading to its children in a dictionary structure indexed by the first letters of the edges' labels. For every leaf,
we store its label and color. An insert operation receives a string (given as a handle from a dynamic collection),
together with its label and color, and should create its corresponding leaf.
A delete operation does not actually remove a leaf, but simply removes its label.
However, in order to not increase the space complexity, we rebuild the whole data
structure from scratch after every $m$ updates. This rebuilding does not incur any extra cost asymptotically; 
the time required for it can be deamortized using standard techniques.

\begin{lemma}\label{lem:tries}
Each update in $\mathcal{R}$ or $\mathcal{B}$ implies $\cO(1)$ updates in $\T_{P}$ and $\T_{Q}$ that can be computed in $\cO(\log n)$ time. 
\end{lemma}
\begin{proof}
Inserting a new leaf, corresponding to string $U$, to $\T_{P}$ requires possibly splitting an edge into two by creating a new explicit
node, and then attaching a new leaf to an explicit node. To implement this efficiently, we maintain the set $C$ of path-labels of explicit nodes of $\T_P$ in a balanced search tree, sorted in lexicographic order.
Using $\lcp$ queries (cf.~\cref{thm:dynst}), we binary search for the longest prefix $U'$ of $U$ that equals the path-label of some implicit or explicit node of $\T_P$.
If this node is explicit, then we attach a leaf to it.
Otherwise, let the successor of~$U'$ in $C$ be the path-label of node $v$. We split the edge $(\textsf{parent}(v),v)$ appropriately and attach a leaf to the newly created node.
This allows us to maintain $\T_{P}$ after
each insert operation in $\cO(\log n)$ time.

For a delete operation, we can access the leaf corresponding
to the deleted string in $\cO(\log n)$ time using the balanced search tree.
\end{proof}

It thus suffices to show a solution for the \textsc{Dynamic Bicolored Trees Problem} that processes each update in polylogarithmic time.

We will maintain a heavy-light decomposition of both $\T_{1}$ and $\T_{2}$. This can be done by using
a standard method of rebuilding as used by Gabow~\cite{Gabow90}. Let $L(u)$ be the number of leaves in
the subtree of $u$, including the leaves without labels, when the subtree was last rebuilt.
Each internal node $u$ of a tree selects at most one child $v$ and the edge
$(u, v)$ is \emph{heavy}. All other edges are \emph{light}.
Maximal sequences of consecutive heavy edges are called \emph{heavy paths}.
The node $r(p)$ closest to the root of the tree is called the \emph{root} of the
heavy path $p$ and the node $e(p)$ furthest from the root of the tree is called
the \emph{end} of the heavy path. The following procedure receives a node $u$
of the tree and recursively rebuilds the heavy paths in its subtree.

\begin{algorithm}[ht]
  \begin{algorithmic}[1]
    \Function{\decompose}{$u$, $r$}
      \Comment{$r$ is the root of the heavy path containing $u$.}
      \State $S \gets \operatorname{children}(u)$
      \State $v \gets \operatorname{argmax}_{v \in S} L(v)$
      \If{$L(v) \ge \frac{5}{6} \cdot L(r)$}
        \State edge $(u, v)$ is heavy
        \State \decompose($v$, $r$)
        \State $S \gets S \setminus \{v\}$
      \EndIf
      \For{$v \in S$}
        \State \decompose($v$, $v$)
      \EndFor
    \EndFunction
  \end{algorithmic}
\end{algorithm}

Every root $u$ of a heavy path maintains the number of insertions
$I(u)$ in its subtree since it was last rebuilt. When $I(u) \ge \frac{1}{6} \cdot L(u)$, we recalculate the values of $L(v)$
for nodes $v$ in the subtree of $u$ and call $\decompose(u, u)$.
This maintains the property that $L(e(p)) \ge \frac{2}{3} L(r(p))$
for each heavy path $p$ and leads to the following.

\begin{proposition}
There are $\cO(\log m)$ heavy paths above any node.
\label{prop:pathlog}
\end{proposition}

As rebuilding a subtree of size $s$ takes $\cO(s)$ time, by a standard potential
argument, we get the following. (The bottleneck is in updating $I(u)$s.)

\begin{lemma}\label{lem:hl}
The heavy-light decompositions of $\T_1$ and $\T_2$ can be maintained in $\cO(\log m)$ amortized time per update.
\end{lemma}

Note that we used $I(u)$ as number of insertions, not the number of operations
and this does not impose any new problem in the analysis.
If we assign a~credit to~the node each time it is deleted, standard amortized
analysis via accounting method shows that we can rebuild a subtree without
exceeding the claimed time with usage of these additional credits.

The main ingredient of our data structure is a collection of additional structures, each storing
a dynamic set of points. Each such point structure sends its current result to a max-heap,
and after each update we return the largest element stored in the heap.
The problem each of these point structures are designed for is the following.

\Dynproblem{Dynamic Best Bichromatic Points}{A multiset of at most $m$ bicolored points from $[m]\times [m]$.}{Insertions and deletions of points from $[m]\times [m]$.}{Return a pair of points $R=(x,y)$ and $B=(x',y')$ such that $R$ is red, $B$ is blue, and $\min(x,x')+\min(y,y')$ is as large as possible.}

We call the pair of points sought in this problem the \emph{best bichromatic pair of points}.
In~\cref{sec:structure}, we show the following result by modifying 2D range trees.

\begin{lemma}\label{lem:ptstructure}
There is a data structure for \textsc{Dynamic Best Bichromatic Points} that processes each update in $\cO(\log^{2} m)$ amortized time.
\end{lemma}

Conceptually, we maintain a point structure for every pair of heavy paths
from $\T_{P}$ and~$\T_{Q}$. However, the total number of points stored in all point structures
at any moment is only $\cO(m\log^{2}m)$ and the empty structures are not actually created.
Consider heavy paths $p$ of $\T_{1}$ and $q$ of $\T_{2}$.
Let $\ell$ be a label such that there are leaves $u$ in the subtree of $r(p)$ in $\T_{1}$
and $v$ in the subtree of $r(q)$ in $\T_{2}$ with the same color and both labeled by $\ell$.
Then, the point structure should contain a point $(x,y)$ with this color, where
$x$ and $y$ are the string-depths of the nodes of $p$ and $q$ containing $u$ and $v$
in their light subtrees, respectively. It can be verified that then the answer extracted from
the point structure is equal to the sought result, assuming that the corresponding pair
of nodes belongs to $p$ and $q$, respectively. It remains to explain how to maintain
this invariant when both trees undergo modifications.

Splitting an edge does not require any changes to the point structures.
Each label appears only once in $\T_{1}$ and $\T_{2}$, and hence by Proposition~\ref{prop:pathlog}
contributes to only $\cO(\log^{2}n)$ point structures. Furthermore, by navigating
the heavy path decompositions we can access these point structures efficiently. This allows
us to implement each deletion in $\cO(\log^4 n)$ amortized time, employing~\cref{lem:ptstructure}. 
To implement the insertions, we need to additionally
explain what to do after rebuilding a subtree of $u$. In this case, we first
remove all points corresponding to leaves in the subtree of $u$, then rebuild the subtree,
 and then proceed to insert points to existing and potentially new point structures. 
As each insertion affects $\cO(\log n)$ heavy paths, it affects $\cO(\log n)$ rebuilding instances.
By the same standard potential
argument as above, each insertion costs $\cO(\log^4 n)$ amortized time per such instance: we add a point, in $\cO(\log^{2} n)$ time, in $\cO(\log^{2} n)$ point structures.
Hence insertions require $\cO(\log^5 n)$ amortized time.

\paragraph{Wrap-up.}
\cref{lem:anchored} reduces our problem to the \textsc{LCP for Two Families of Pairs of Strings} problem for sets $\mathcal{R}$ and $\mathcal{B}$ of size $\cOtilde(n)$, 
so that each edit operation in $S$ or $T$ yields $\cO(\log^2 n)$ updates to $\mathcal{R}$ and $\mathcal{B}$, which can be performed in $\cO(\log^3 n)$ time due to~\cref{lem:Jup}.
The \textsc{LCP for Two Families of Pairs of Strings} problem is then reduced to the \textsc{Dynamic Bicolored Trees Problem} for trees $\T_1$ and $\T_2$ of size $\cOtilde(n)$, so that each update in $\mathcal{R}$ or $\mathcal{B}$ yields $\cO(1)$ updates to the trees, which can be computed in $\cO(\log n)$ time (\cref{lem:tries}).
We solve the latter problem by maintaining a heavy-light decomposition of each of the trees in $\cO(\log n)$ amortized time per update (\cref{lem:hl}),
and an instance of a data structure for the \textsc{Dynamic Best Bichromatic Points} problem for each pair of heavy paths.
For each update to the trees, we spend $\cO(\log^5 n)$ amortized time to update the point structures.
Thus, each update in one of the strings costs a total of $\cO(\log^7 n)$ amortized time.
In order to keep the space usage of our data structure $\cOtilde(n)$, we reinitialize it after every $\Theta(n)$ updates.

\subsection{Dynamic Best Bichromatic Points}
\label{sec:structure}

In this section we prove~\cref{lem:ptstructure}, i.e., design an efficient data structure for the \textsc{Dynamic Best Bichromatic Points} problem.

First, we show that we can assume that all $x$ and $y$ coordinates of points are distinct.
Let us replace identical points of the same colour with a single point, with which we store its multiplicity as satellite information.
Then, we perform the following standard perturbation. 
Namely, we can (implicitly) replace each red (resp.~blue) point $(x,y)$ with $((x|y|0),(y|x|0))$ (resp.~$((x|y|1),(y|x|1))$),
and use the lexicographic order to perform comparisons for each coordinate (cf~\cite[Section 5.5]{DBLP:books/lib/BergCKO08}). 
As the data structures that we use are comparison based, the above transformation does not affect the complexities.

We maintain an augmented dynamic 2D range tree~\cite{WillardL85} over the multiset of points.
This is a balanced search tree $\T$ (called primary) over the $x$ coordinates of all points
in the multiset in which every $x$ coordinate corresponds to a leaf and, more generally, every
node $u\in \T$ corresponds to a range of $x$ coordinates denoted by $x(u)$.
Additionally, every $u\in \T$ stores another balanced search tree $\T_{u}$ (called secondary)
over the $y$ coordinates of all points $(x,y)\in S$ such that $x\in x(u)$.
Thus, the leaves of $\T_{u}$ correspond to $y$ coordinates of such points,
and every $v\in \T_{u}$ corresponds to a range of $y$ coordinates denoted by $y(v)$.
We interpret every $v\in \T_{u}$ as the rectangular region of the plane $x(u) \times y(v)$,
and, in particular, each leaf $v\in \T_{u}$ corresponds to a single point in the multiset.
Each node $v\in \T_{u}$ will be augmented with some extra information that
can be computed in constant time from the extra information stored in its children. Similarly,
each node $u\in \T$ will be augmented with some extra information that
can be computed in constant time from the extra information stored in its children together
with the extra information stored in the root of the secondary tree $\T_{u}$.
The complexity of the described approach depends on the balanced tree that is used to represent the primary and secondary tree.
Irrespectively of what this extra information is, as explained by Willard and Lueker~\cite{WillardL85},
if we implement the primary tree as a $BB(\alpha)$ tree and each secondary tree as a balanced search
tree, each insertion and deletion can be implemented in $\cO(\log ^{2}m)$
amortized time.

Before we explain what is the extra information, we need the following notion.
Consider a non-leaf node $u\in \T$ and let $u_{\ell},u_{r}\in \T$ be its children.
Let $v\in \T_{u}$ be a non-leaf node with children $v_{\ell},v_{r}\in \T_{u}$.
The regions $A=x(u_{\ell})\times y(v_{\ell})$, $B=x(u_{\ell})\times y(v_{r})$,
$C=x(u_{r})\times y(v_{\ell})$ and $D=x(u_{r})\times y(v_{r})$ partition $x(u)\times y(v)$
into four parts. We say that two points $p=(x,y)$ and $q=(x',y')$ with $x<x'$
are shattered by $v\in \T_{u}$ if and only if $p\in A$ and $q\in D$ or $p\in B$ and $q\in C$
(note that the former is only possible when $y<y'$ while the latter can only hold when $y>y'$).

\begin{proposition}
Any pair of points in the multiset is shattered by a unique $v\in \T_{u}$ (for a unique $u$).
\end{proposition}

\FIGURE{t}{0.9}{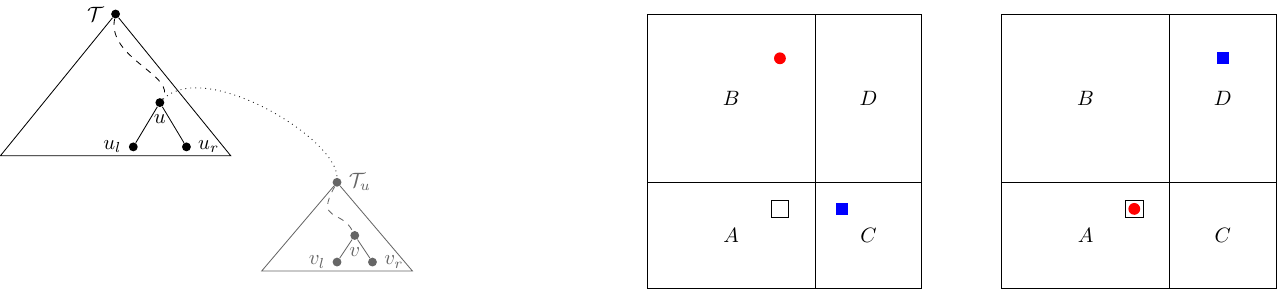}{
  Left: A 2D range tree.
  Right: Node representing regions $A$, $B$, $C$, $D$.
  The best pair for each case is denoted by a~small square.
}

Now we are ready to describe the extra information. Each node $u\in \T$ stores
the best bichromatic pair with $x$ coordinates from $x(u)$. Each node $v\in \T_{u}$
stores the best bichromatic pair shattered by one of its descendants $v'\in \T_{u}$ (possibly $v$ itself).
Additionally, each node $v\in \T_{u}$ stores the following information about points
of each color in its region:
\begin{enumerate}
  \item the point with the maximum $x$,
  \item the point with the maximum $y$,
  \item a point with the maximum $x+y$.
\end{enumerate}
We need to verify that such extra information can be indeed computed in constant time
from the extra information stored in the children.

\begin{lemma}
Let $v\in \T_{u}$ be a non-leaf node, and $v_{\ell},v_{r}$ be its children. The extra
information of $v$ can be computed in constant time given the extra information
stored in $v_{\ell}$ and $v_{r}$.
\end{lemma}
\begin{proof}
This is clear for the maximum $x$, $y$ and $x+y$ of each color, as we can take the maximum
of the corresponding values stored in the children. For the best bichromatic pair shattered
by a descendant $v'$ of $v$, we start with considering the best bichromatic pair shattered
by a descendant $v'_{\ell}$ of $v_{\ell}$ and $v'_{r}$ of $v_{r}$. The remaining case is
that the best bichromatic pair is shattered by $v$ itself. Let $A,B,C,D$ be as in the definition
of shattering. Without losing generality we assume that the sought pair is $p=(x,y)$ and $q=(x',y')$
with $x<x'$, red $p$ and blue $q$. We consider two cases:
\begin{enumerate}
\item $p\in A$ and $q\in D$: the best such pair is obtained by taking $p$ with the maximum $x+y$ and any $q$,
\item $p\in B$ and $q\in C$: the best such pair is obtained by taking $p$ with the maximum $x$
and $q$ with the maximum $y$.
\end{enumerate}
In both cases, we are able to compute the best bichromatic pair shattered by $v$ using the extra
information stored at the children of $v$. See Figure~\ref{fig:primary-secondary}.
\end{proof}

\begin{lemma}
Let $u\in \T$ be a non-leaf node, and $u_{\ell},u_{r}$ be its children. The extra
information of $v$ can be computed in constant time given the extra information
stored in $v_{\ell}$, $v_{r}$ and the root of $\T_{u}$.
\end{lemma}
\begin{proof}
We seek the best bichromatic pair with $x$ coordinates from $x(u)$. If the $x$ coordinates
are in fact from $x(u_{\ell})$ or $x(u_{r})$, we obtain the pair from the children of $u$.
Otherwise, the pair must be shattered by some $v\in \T_{u}$ that is a descendant
of the root of $\T_{u}$, so we obtain the pair from the root of $\T_{u}$.
\end{proof}

\section{Lower Bounds}
\label{sec:lowerbound}

We show lower bounds in the cell-probe model of computation, perhaps the strongest model for proving lower bounds, which only accounts for memory accesses and neglects the actual computation time.
We make the standard assumption that the size of each machine word is $\Theta(\log n)$, where $n$ is an upper bound on the size of the input at all times.
Note that this assumption is necessary if we want to be able to specify memory locations storing parts of the input in $\cO(1)$ time.

\subsection{A Lower Bound for Partially Dynamic LCS under Edit Operations}\label{sec:lb_ps}

\newcommand{\report}{\textsf{report}}
\newcommand{\parity}{\textsf{parity}}

In this subsection, we will show a lower bound for the partially dynamic LCS problem, when edit operations are allowed, by a reduction from a following variant of the dynamic partial sums problem.
Fredman and Saks~\cite{10.1145/73007.73040} showed that any data structure for the \textsc{Dynamic Partial Sums Modulo 2} (DPSM2) problem, specified below, cannot handle both updates and queries in $o(\log n/ \log\log n)$ time under our assumption on the word-size. 

\Dynproblem{Dynamic Partial Sums Modulo 2}{An integer array $A$ of size $n$.}{$A[k]:=A[k]+\delta$.}{Given $i$, return $\left(\sum_{j\leq i} A[j]\right) \bmod 2$.}

Later, Husfeldt and Rauhe~\cite{DBLP:journals/siamcomp/HusfeldtR03} considered different variants of the above problem, including the following \emph{promise variant}.

\Dynproblem{Promise DPSM2}{An integer array $A$ of size $n$, with all entries in $\{-1,0,1\}$.}{$A[k]:=x$ for $x\in \{-1,0,1\}$.}{$\parity_A(i,s)$: Given $i$ and $s$, with the promise that $|\sum_{j\leq i} A[j] - s|\leq 1$, return $\left(\sum_{j\leq i} A[j]\right) \bmod 2$.}

In particular, Husfeldt and Rauhe showed that any data structure for the \textsc{Promise DPSM2} problem cannot process both updates and queries in $o(\log n/ \log\log n)$ time.

Further, DPSM2 reduces to the following problem~\cite{10.1145/73007.73040}.

\Dynproblem{List Representation}{A list $L$ of size at most $n$ over universe $[n]$.}{``Delete the $k$-th element of $L$'' or ``insert $a \in [n]$ after the $k$-th element of $L$''.}{Report the $i$-th element of $L$.}

We will slightly abuse notation, and refer to the $i$-th element of $L$ as $L[i]$.
In particular, we consider a verification variant of this problem, where the allowed queries are of the following type: check whether the $i$-th element of $L$ is $x$; we denote such a query by $\report_L(i,x)$. 
In our proof below, we provide a reduction from \textsc{Promise DPSM2} to the verification variant of list representation; the reduction for the standard variants was omitted in~\cite{10.1145/73007.73040}.

Further note, that by interpreting list $L$ as a string, one can see that any data structure, for any problem on dynamic strings under edit operations, that enables one to retrieve any letter of the string, specified by its current position, with a constant number of queries and updates, cannot process both queries and updates in $o(\log n/ \log\log n)$ worst-case time deterministically.
We show that this is indeed the case for partially dynamic LCS under edit operations.

\begin{theorem}\label{lb:indels}
Any data structure for maintaining an LCS of  a static string $T$ and a dynamic string $S$ that undergoes edit operations,
each of length at most $n$, requires $\Omega(\log n/ \log \log n)$ time per update operation.
\end{theorem}
\begin{proof}
Consider an instance of the \textsc{Promise DPSM2} problem over an array $A$ of size $n$.

We define an instance of the verification variant of the list representation problem, of size $m=10n^2$, with list $L$ initially consisting in $n$ copies of $1^{5n} 2 3 \cdots (5n)(5n+1)$.
The only update operations that we will perform on $L$ are deletions and insertions of $1$s.
More specifically, we will maintain the following invariant: the $k$-th $2$ in $L$ is preceded by
\begin{enumerate}
\item a (maximal) block of $5n-1$ $1$s if $A[k] = -1$;
\item a (maximal) block of $5n$ $1$s if $A[k] = 0$;
\item a (maximal) block of $5n+1$ $1$s if $A[k] = 1$.
\end{enumerate}
This is ensured for the initial array $A$ by changing the initial list $L$.
As the number of $1$s in each block can differ by at most $2$, depending on value of $A_k$, thanks to having at least $5n-1$ $1$s in each block, positions of the boundaries of these blocks can be easily approximated.

To maintain the invariant, upon an update of $A[k]$, we make at most two insertions of $1$s or at most two deletions of $1$s in the block of $1$s that precedes the $k$-th $2$ in $L$.
By the invariant, at any point our list can be obtained by our initial list by at most $n$ insertions and deletions of $1$s.
Consequently
$L[10kn+3n-1]L[10kn+3n]$ is always part of the sought block of $1$s and hence we can perform our insertions or deletions there.

We now show that a $\parity_A(i,s)$ query reduces to a constant number of $\report_L(\cdot,\cdot)$ queries.
Similarly to before, the element $L[10in+8n-1]$ lies always strictly between the $i$-th and the $(i+1)$-st blocks of $1$s.
Note that when $L$ was initialized, we had $L[10in+8n-1]=3n$. Every deletion of a $1$ to its left increments $L[10in+8n-1]$, while every insertion of a $1$ to its left decrements $L[10in+8n-1]$, and updates to its right do not change its value.
The promise $s$ is an additive 1-approximation of the following quantity: the total length of the first $i$ blocks of~$1$s minus $5in$.
It then suffices to refine this approximation, as this directly yields the parity of the sought prefix sums.
We can achieve this by performing queries $\report_L(10in+8n-1,3n+x)$ for all $x \in \{s-1,s,s+1\}$.

We have reduced \textsc{Promise DPSM2} problem over an array $A$ of size $n$ to
the verification variant of the list representation problem of size $m$, with each query and updates in the original instance
being translated to a constant number of queries and updates in the new instance.
By the lower bound of Husfeld and Rauhe~\cite{DBLP:journals/siamcomp/HusfeldtR03}, any data structure for the former problem
cannot process both queries and updates in $o(\log n/\log\log n)$ time.
Consequently, any data structure for the latter problem cannot process both queries and updates
in $o(\log n/\log\log n)=o(\log m/\log\log m)$ time.

We now proceed to reduce the verification variant of the list representation problem of size $m$
to the partially dynamic LCS problem under edit operations for strings of total length at most $6m$.
When combined with the above reduction, we will be able to conclude that any data structure for the latter
cannot process queries and updates in $o(\log m/\log\log m)$ time.

Our static string is \[T=11\#22\# \cdots m \#.\]
We will maintain dynamic string \[S=L[1]\$\$L[2]\$\$ \cdots L[|L|] \$\$.\]
Each update in $L$ naturally translates to at most three updates in $S$.

We process a query $\report_L(i,x)$ as follows. 
We issue update $S[3i+2]:=x$ to our LCS data structure.
Then, it is readily verified that the LCS of $S$ and $T$ is of length $2$ if and only if $L[i]=x$.
We then revert our changes, setting $S[3i+2]:=\$$.
\end{proof}

\subsection{Lower Bounds for Dynamic LCS under Substitution Operations}

We now proceed to show that the studied problems are hard even when the only allowed operations are substitutions.
The lower bounds will be obtained by a series of reductions. We need a few auxiliary definitions.

We call $(b,n)$-normal the tree obtained from a complete $b$-ary tree $\T$ of depth $d=\log_b n$ by attaching at most
$b^{d+1-d'}$ extra leaves to every node at depth $d'$, for $d'=1,2,\ldots,d$. (The root is at depth $0$.)
The size of such a tree is $\sum_{d'=1}^{d}b^{d'}\cdot b^{d+1-d'} = \cO(n\log n)$.
To avoid clutter, we will assume that $\log_{b}n$ is an integer.

\DSproblem{\textsc{Restricted HIA}}{Two $(b,n)$-normal trees with some pairs of extra leaves attached to
nodes at depths $d'$ and $d+1-d'$,
for some $d'=1,2,\ldots,d$, having the same unique label.}{Given two leaves at depth $d$, is the total weight of their HIA $d+1$?}

\noindent
Observe that in an instance of \textsc{Restricted HIA} built for $\T_{1}$ and $\T_{2}$ there is no need to attach
more than one pair of extra leaves to the same pair of nodes. Additionally, a query concerning a leaf $u\in \T_{1}$
and $v\in \T_{2}$ can be reformulated as seeking an ancestor $u'$ of $u$ at depth $d'$ and an ancestor $v'$ of $v$
at depth $d+1-d'$ such that there is a pair of extra leaves with the same unique label attached to both $u'$ and $v'$,
for some $d'=1,2,\ldots,d$.

The butterfly graph of degree $b$ and depth $d$ has vertex set $V=\{(i,\vec{r}):i=0,\ldots,d,\vec{r} \in [b]^d \}$.
A vertex $(i,\vec{r})$, $i=0,1,\ldots,d-1$, has outgoing edges to all vertices $(i+1,\vec{r'})$ such that $\vec{r}$ and~$\vec{r'}$ differ
only on the $(i+1)$-th coordinate -- there are $b$ such vertices.
We call vertices $(0,\vec{r})$ sources and vertices $(d,\vec{r})$ sinks. We actually manipulate
the vertex labels as numbers, but it is conceptually easier to think of them as vectors. Note that each source can reach each
sink by a unique path corresponding to transforming the source's vector into the sink's vector coordinate by coordinate.

\DSproblem{\textsc{Butterfly Reachability}}{A subgraph $H$ of a butterfly graph $G$.}{Is sink $v$ reachable from from source $u$?}

\begin{theorem}[\cite{DBLP:journals/siamcomp/Patrascu11}]
Any structure of size $\cOtilde(n)$ for \textsc{Butterfly Reachability}, where $n$ is the size of the butterfly graph, requires
query time $\Omega(\frac{\log n}{\log\log n})$.
\label{thm:butterfly}
\end{theorem}

The reduction in the proof of the following lemma closely resembles the reduction from \textsc{Butterfly Reachability} to
range stabbing in 2D from~\cite{DBLP:journals/siamcomp/Patrascu11}, but we provide the details for completeness.

\begin{lemma}
If there is a structure of size $\cOtilde(n)$ for \textsc{Restricted HIA} that answers queries in $t$ time, then there is
a structure of size $\cOtilde(n)$ for \textsc{Butterfly Reachability} that answers queries in time $\cO(t)$.
\label{lem:restricted}
\end{lemma}
\begin{proof}
We will show how to reduce answering a source-to-sink reachability query on a subgraph of a butterfly graph of degree $b$, depth $d$ and size $n$
to answering a HIA query on two $(b,n)$-normal trees $\T_{1}$ and $\T_{2}$.

Let $\T_1$ (resp.~$\T_2$) be the complete $b$-ary tree over vertices of level $0$ (resp.~$d$), sorted
lexicographically (resp.~sorted by the lexicographic order of their reverses). In $\T_1$, the ancestor of a leaf $(0,\vec{r})$
at depth $\delta$ is an ancestor of any leaf $(0,(*,\ldots,*,r_{d-\delta+1},\ldots ,r_d))$, where $*$ indicates an arbitrary
value. In $\T_2$, the ancestor of a leaf $(d,\vec{r})$ at depth $\delta$ is an ancestor of any leaf $(d,(r_1,\ldots,r_{\delta},*,\ldots ,*))$.

In the butterfly graph, the edge from vertex $(i,\vec{r})$ to vertex $(i+1,\vec{p})$ in $G$ is on the unique path from
each the source $(0,(*,\ldots,*,r_{i+1},\ldots ,r_d))$ to each sink $(d,(r_1,\ldots,r_{i},p_{i+1},*,\ldots ,*))$.
For each missing edge $((i,\vec{r}),(i+1,\vec{p}))$, we attach a pair of leaves with the same unique identifier:
\begin{enumerate}
	\item in $\T_1$ at the ancestor of $(0,\vec{r})$ at depth $d-i$,
	\item in $\T_2$ at the ancestor of $(d,\vec{p})$ at depth $i+1$.
\end{enumerate}
Note that the sum of the depths of these nodes is $d+1$. Consequently, sink $v$ is not reachable from source $u$
if and only if the total weight of the HIA of $u\in \T_{1}$ and $v\in \T_{2}$ is $d+1$.
Both constructed trees are indeed $(b,n)$-normal, as the number of leaves attached to a node at depth $d'$
is at most $b^{d+1-d'}$, for every $d'=1,2,\ldots,d$ (and there are no leaves attached to the roots).
\end{proof}

\subsubsection{Partially Dynamic LCS}

\begin{lemma}
\label{lem:fragments}
Given an instance of \textsc{Restricted HIA} on $(b,n)$-normal trees, we can construct
$\cO(\sqrt{n})$ instances of \textsc{Restricted HIA} on $(b,\sqrt{n})$-normal trees such that
any query to the original instance can be reduced in constant time to two queries to the
smaller instances using a preprocessed table of size $\cO(n)$.
\end{lemma}

\begin{proof}
Recall that $\T_{1}$ and $\T_{2}$ in the \textsc{Restricted HIA} problem are complete $b$-ary trees of depth $d=\log_{b}n$ with some extra leaves.
We want to partition both complete $b$-ary trees into smaller edge-disjoint trees of depth $d/2$ called fragments.
This is done by selecting all nodes at depth $d/2$ in both trees. For every selected node $w$, we create
two copies of it, $w'$ and $w''$, such that $w'$ inherits the ingoing edge from $w$'s parent and becomes a leaf
of a fragment, while $w''$ inherits the outgoing edges to children of $w$ and becomes the root of some fragment.
We will say that the new leaf $w'$ corresponds to the original $w$, and every node that was not split corresponds to itself.
Thus, we partition each tree into the top fragment and $\sqrt{n}$ bottom fragments.
We create an instance of \textsc{Restricted HIA} on $(b,\sqrt{n})$-normal trees for the following pairs of fragments:
the top fragment of $\T_{1}$ together with each bottom fragment of $\T_{2}$, and similarly
each bottom fragment of $\T_{1}$ together with the top fragment $\T_{2}$.
We need to specify how to attach the extra leaves in each instance.

Consider the top fragment $A$ of $\T_{1}$ and a bottom fragment $B$ of $\T_{2}$. 
Take a pair of leaves $u$ and $v$ with the same unique label attached to the nodes corresponding
to $u'\in A$ and $v'\in B$ in the original instance. We attach a pair of such leaves to $u'$ and $v'$ in $(A,B)$. We claim
that $(A,B)$ is a valid instance of \textsc{Restricted HIA} over $(b,\sqrt{n})$-normal trees.
This requires checking that the number of leaves attached to a node at depth $d'$ in $A$ or $B$ is at most $b^{d/2+1-d'}$.
In the original instance, a pair of leaves is attached at depths $d'$ and $d+1-d'$, for some $d'=1,2,\ldots,d/2$.
In $(A,B)$ the corresponding pair of leaves is attached at depths $d'$ and $d/2+1-d'$. 
For any node $u$ at depth $d'$ in $A$, we have at most $b^{d/2+1-d'}$ nodes at depth $d/2+1-d'$ in $B$,
and hence attach that many leaves to $u$. Similarly, for any node $v$ at depth $d/2+1-d'$ in $B$,
we have at most~$b^{d'}$ nodes at depth $d'$ in $A$, and hence attach that many leaves to $v$.
Consequently, $(A,B)$ is a valid instance of \textsc{Restricted HIA} over $(b,\sqrt{n})$-normal trees.
A symmetric argument applies for the case when $A$ is a bottom fragment of $\T_{1}$ while $B$ is the top
fragment of $\T_{2}$.

We have obtained $\cO(\sqrt{n})$ instances of \textsc{Restricted HIA} over $(b,\sqrt{n})$-normal trees.
Consider a query that seeks an ancestor $u'$ at depth $d'$ of a leaf $u\in \T_{1}$ and an ancestor $v'$ at
depth $d+1-d'$ of a leaf $v\in \T_{2}$ such that there is a pair of leaves with the same unique label attached
to both $u'$ and $v'$, for some $d'=1,2,\ldots,d$. For $d' \leq d/2$ this can be retrieved by querying
the instance $(A,B)$, where $A$ is the top fragment of $\T_{1}$ and $B$ is a bottom fragment of $\T_{2}$
containing the leaf corresponding to $v$. To translate the query, we need to retrieve the first ancestor
of $u$ for which the corresponding node belongs to the top fragment, but this can be simply preprocessed
and stored for every leaf. Similarly, for $d' > d/2$ we should query the instance $(A,B)$, where $A$ is
a bottom fragment of $\T_{1}$ and $B$ is the top fragment of $\T_{2}$. Thus, a query to the original
instance reduces to two queries to the smaller instances, and the translation takes constant time
after $\cO(n)$-space preprocessing.
\end{proof}

\begin{lemma}
If there exists a structure of size $\cOtilde(n)$ for maintaining the LCS of a dynamic string $S$ and a static string $T$,
each of length $\cO(n\log^{2}n)$, requiring $t$ time per update, then there exists a structure of size $\cOtilde(n)$
for \textsc{Restricted HIA} on $(b,n)$-normal trees that answers queries in time $\cO(t)$.
\label{lem:lcs2}
\end{lemma}

\begin{proof}
We first apply Lemma~\ref{lem:fragments} to obtain $\cO(\sqrt{n})$ instances of \textsc{Restricted HIA}
on $(b,\sqrt{n})$-normal trees. In the $i$-th instance we are working with two $(b,\sqrt{n})$-normal trees
$\T_{1}^{i}$ and $\T_{2}^{i}$ obtained from complete $b$-ary trees of depth $d=\log_{b}\sqrt{n}$.
The edges outgoing from a node in a complete $b$-ary tree
can be naturally identified with numbers $1,2,\ldots,b$. This allows us to uniquely represent every node $w$ of such
a tree by the concatenation of the numbers corresponding to the edges on the path from $w$ to the root, denoted by
$\pathtree(w)$. Similarly, $\pathtree^{r}(w)$ is the concatenation of the numbers corresponding to the edges on the path from
the root to $w$.

For every $i$, for every $d'=1,2,\ldots,d$ we consider every pair of leaves $u\in \T_{1}^{i}$ and $v\in \T_{2}^{i}$ with the same label
attached to $u'\in \T_{1}^{i}$ and $v'\in \T_{2}^{i}$
at depths $d'$ and $d+1-d'$, respectively, and construct the following gadget:
\[ \pathtree(u') \mathtt{\$_{i}}  \pathtree^{r}(v'). \]
We concatenate all such gadgets, separated by $\texttt{.}$ characters, to obtain $T$. The length of $T$ is $\cO(n\log^{2}n)$.

For each pair of strings $s$ and $t$ of length $d$ over $\{1,2,\ldots,b\}$ we construct the following gadget:
\[ s \texttt{\#} t \]
We concatenate all such gadgets, separated by $\texttt{,}$ characters, to obtain $S$. The length of $S$ is $\cO(n\log n)$.

A query to the original instance of \textsc{Restricted HIA} can be translated in constant time into two queries to the smaller instances.
To answer a query concerning a pair of leaves $u\in \T_{1}^{i}$ and $v\in \T_{2}^{i}$, we locate the gadget corresponding
to $\pathtree(u) \# \pathtree^{r}(v)$ in $S$. This can be done in constant time by enumerating the strings
$s$ and $t$ in the natural lexicographical order and storing the rank of each leaf. Then, we temporarily
replace $\#$ by $\$_{i}$, find the length $L$ of the LCS, and then restore the gadget.
We claim that $L=d+2$ if and only if the total weight of HIA is $d+1$. Recall that the total weight of HIA is $d+1$
if and only if there exists an ancestor $u'$ of $u$ at depth $d'$ and $v'$ of $v$ at depth $d+1-d'$ with
a pair of leaves having the same unique label. Since gadgets are separated with different characters in $S$ and $T$,
any $\pathtree(w)$ is of length at most $d$, $L=d+2$ must correspond to a substring $\pathtree(u') \$_{i} \pathtree(v')$,
for some ancestor $u'$ of~$u$ at depth $d'$ and $v'$ of $v$ at depth $d+1-d'$. But such a substring occurs in $T$
if and only if a pair of leaves connected to $u'$ and $v'$ has been activated, so our answer is indeed correct.
\end{proof}

\begin{theorem}\label{thm:lbpartially}
Any structure of $\cOtilde(n)$ size for maintaining an LCS of a dynamic string $S$ and a static string $T$,
each of length at most $n$, requires $\Omega(\log n/ \log \log n)$ time per update operation.
\end{theorem}

\paragraph{Lower bound for amortized update time.}
We next show that the above theorem also holds when amortization is allowed.
To this end, let us suppose that there exists an $\cOtilde(n)$-size data structure that processes each update in amortized time upper bounded by $t=o(\log n / \log \log n)$.
We will show that this implies the existence of an $\cOtilde(n)$-size data structure with $\cO(t)$ worst-case update time, contradicting~\cref{thm:lbpartially}.

First, let us observe that the proof of~\cref{lem:lcs2} actually considers a restricted variant of the partially dynamic LCS problem, where updates come in pairs:
we temporarily replace some $\#$ by $\$_{i}$ and then revert this change.
We initialize the data structure that achieves the claimed space and amortized update time complexities for the static string $T$ and a dynamic string of length $|S|$.
Given any sequence of $m$ updates, the total time that this ``blank'' data structure requires to process them is at most $m \cdot t$.
We first perform at most $n$ updates to transform the dynamic string to $S$.
At this point, the budget towards future (expensive) updates can be no more than $n \cdot t$.
While we can, we consider a pair of updates of the desired type, i.e., replace some $\#$ by $\$_{i}$ and then revert this change, with total cost at least $3t$.
This can be done at most $n$ times, else the amortized cost over all our updates would exceed $t$.
After having reached a point where such a pair does not exist, every (pair of) update(s) we might perform costs $\cO(t)$ worst-case time.
From then on, when executing a pair of updates we store the modified memory locations together with their original content.
This allows us to restore the structure to the state before the current pair of updates, and so any possible pair of updates
will take only $\cO(t)$ worst-case time.
To bound the size of the structure, recall that we have assumed that the structure takes $\cOtilde(n)$ space.
Additionally, we perform $\cO(n)$ updates after the after the initialization in order to obtain the version of
the data structure in which all pairs of updates are cheap, which adds $\cO(n\cdot t)$ to the size of the structure,
which remains to be $\cOtilde(n)$.

\paragraph{Allowing Las Vegas randomization.}
Let us first note that~\cref{thm:butterfly} holds even when Monte Carlo randomization is allowed.
This is shown in~\cite{DBLP:journals/siamcomp/Patrascu11} through a reduction from the \textsc{Lopsided Set Disjointness} (LSD) problem, which we define next.

\genproblem{Lopsided Set Disjointness}{Alice and Bob receive sets $A$ and $B$, respectively, over a universe $U$ such that $|U|/|A|=b$.}{Is $A\cap B=\emptyset$?}

\begin{theorem}[\cite{DBLP:journals/siamcomp/Patrascu11}]\label{thm:lsdlb}
Fix $\delta>0$. If a protocol for \textsc{Lopsided Set Disjointness} has error less than $\frac{1}{9999}$, then either Alice sends at least $\delta |A| \log b$ bits or Bob sends at least $|A|\cdot b^{1-\cO(\delta)}$ bits.
\end{theorem}

Let us suppose that there exists a data structure for partially dynamic LCS that processes each update in expected amortized time $t$.
Formally, the expected total time to process the first $n$ operations should be at most $n\cdot t$, for every $n$.
Our aim is to perform some updates on a blank data structure to ensure that the next update takes $\cO(t)$ expected worst-case time.

As in our current deamortization proof for this problem, we first initialize the structure using $\cO(n)$ updates. Then, while this is possible, we perform an expensive (pair of) update(s). 
However, now expensive means that the expected time is larger than $C\cdot t$ for some constant $C$.
Note that in order to compute the expected time we consider all possible sequences of
random bits. For each such sequence, we feed the structure with the whole sequence of updates chosen so far, and
consider the time taken by the next update. By taking the average over all sequences of random bits we get the expected time
of an update.
By the same argument as above, there are $\cO(n)$ such updates.
Hence, the next update will require $\cO(t)$ expected worst-case time.

Then, by combining~\cref{lem:restricted,lem:lcs2} we get a distribution over $\cOtilde(n)$-size data structures for
the \textsc{Butterfly Reachability} problem such that any query takes $\cO(t)$ expected worst-case time.
Note that we can assume that the query itself is deterministic once we have fixed the data structure.
The rest of this section is devoted to proving the following result, which is an extension of~\cref{thm:butterfly}.

\begin{theorem}
Suppose that there exists a probability distribution over $\cOtilde(n)$-size data structures for the \textsc{Butterfly Reachability} problem,
so that the expected (over the distribution) time to answer each reachability query is at most $t$.
Then $t=\Omega(\frac{\log n}{\log\log n})$.
\label{thm:butterflyrand}
\end{theorem}

\begin{proof}
We will show that a probability distribution over $\cOtilde(n)$-size data structure for \textsc{Butterfly Reachability} with 
the expected query time being $o(\frac{\log n}{\log\log n})$ implies a protocol for LSD that contradicts~\cref{thm:lsdlb}.

Let us first briefly recall how a lower bound for a communication problem can yield a lower bound for a data-structure problem (as in~\cite{DBLP:journals/siamcomp/Patrascu11}).
Bob constructs a data structure over his set $B$, while Alice simulates a batch of $k$ queries, asking Bob for each of the memory locations of the data structure accessed by those queries.
If the data structure of Bob has size $s$, for each cell-probe, Alice sends $\log \binom{s}{k}$ and Bob responds with $k\cdot w$ bits.
Then, if the data structure answers queries in time $t$, we get a protocol for LSD where Alice sends $t\cdot \log \binom{s}{k}$ bits and Bob sends $t\cdot k\cdot w$ bits, where $w=\Theta(\log s)$ is the word-size.

Our proof closely follows the proof of \cite[Reduction 12]{DBLP:journals/siamcomp/Patrascu11}.
This proof shows how a structured version of LSD (for which~\cref{thm:lsdlb} still holds), reduces to the problem of answering in parallel $|A|/d$ reachability queries for a subgraph of a degree-$b$ butterfly graph with $|A|=n$ non-sink vertices, $n\cdot b$ edges, and depth $d=\Theta(\log n / \log b)$.
We set $b=2\log^{2}n$, so $d=\Theta(\log n/\log\log n)$, the size of our structure constructed for the graph is $s=\cOtilde(n)$, and $w=\Theta(\log n)$.

Bob chooses from the probability distribution a structure of size $s$ constructed for the graph,
and Alice simulates the queries, allowing $c\cdot t$ time for each of them, where $c$ is a constant that will be specified later.
By Markov's inequality, each query requires more than $c\cdot t$ time with probability at most $1/c$.
Hence, in expectation, $\frac{|A|}{dc}$ of Alice's $|A|/d$ queries will remain unanswered.
Then, we can apply again Markov's inequality, to show that the probability of having more than $9999\cdot \frac{|A|}{dc}$ such queries is at most $1/9999$.
In this case, the protocol fails, which is allowed by~\cref{thm:lsdlb}.

For each query that still remains unanswered (out of the allowed $9999\cdot \frac{|A|}{dc}$), Alice explicitly sends the source and sink to Bob, using $2\log s$ bits, 
while Bob checks reachability and responds with 1 bit.
Thus, overall, Alice sends at most 
\begin{align*}
c t\cdot \log \binom{s}{|A|/d} + 9999\cdot 2\log s \cdot \frac{|A|}{dc} & =  c\cdot|A|\cdot \cO \left( \frac{t}{d}\cdot \log \left( \frac{s\cdot d}{|A|} \right) \right) + \frac{|A|}{c} \cdot \cO \left( \frac{\log s}{d} \right) \\
& =  c\cdot n\cdot \cO \left( \frac{t}{d} \cdot \log\log n\right) + \frac{n}{c} \cdot \cO(\log \log n)
\end{align*}
bits and Bob sends at most $c\cdot t \cdot w \cdot  \frac{|A|}{d}+9999\cdot \frac{|A|}{dc} = c\cdot n \cdot \frac{t}{d}\cdot\log n+\cO(\frac{n}{cd})$ bits.

We employ~\cref{thm:lsdlb} with $\delta$ small enough as to make the $b^{1-\cO(\delta)}$ term in its statement at least $\sqrt{b}$.
Then, either Alice sends at least $\delta n \log b$ bits or Bob sends at least $n \sqrt{b}$ bits.
\begin{itemize}
\item In the former case, by setting $c$ to be sufficiently large so that the second term in the expression for the bits sent by Alice is at most $\delta/2\cdot n \log b$, we have $\frac{t}{d} \cdot \log\log n = \Omega(\log b)$, so $\frac{t}{d}=\Omega(1)$.
\item In the latter case, again by setting $c$ to be sufficiently large so that the second term in the expression for the bits sent by Bob
is at most $\frac{1}{2} \cdot n\sqrt{b}$, we have $\frac{t}{d} \cdot \log n = \Omega(\sqrt{b})$, so $\frac{t}{d}=\Omega(1)$.
\end{itemize}
In either case, we obtain $t=\Omega(d)=\Omega(\log n / \log \log n)$.
\end{proof}

\subsubsection{Fully Dynamic LCS}

We now proceed to dynamic LCS. This requires an additional definition.

\Dynproblem{Activation HIA}{Two unlabeled $(b,n)$-normal trees.}{Activate a pair of leaves, one from each tree, attached at depths $d'$ and $d+1-d'$,
for some $d'=1,2,\ldots,d$, by giving them a unique label.}{Given two leaves at depth $d$, is the total weight of their HIA $d+1$?}

\begin{lemma}
If, for some constant $c$, there exists a structure of size $\cO(n^c)$ for \textsc{Activation HIA}
that processes updates in time $t_{u}$ and answers queries in time $t_{q}$, then there exists a structure of size $\cOtilde(n\cdot t_{u})$ for \textsc{Restricted HIA}, which answers queries in time $\cO(t_{q})$.
\label{lem:activation}
\end{lemma}
\begin{proof}
Recall that $\T_{1}$ and $\T_{2}$ in the \textsc{Restricted HIA} problem are complete $b$-ary trees of depth $d=\log_{b}n$ with some extra leaves.
Generalizing the proof of Lemma~\ref{lem:fragments},
we want to partition both complete $b$-ary trees into smaller edge-disjoint trees of depth $D=d/c$ called fragments.
Let us first modify $\T_2$ by attaching a path consisting of $p$ edges, where $p <D$ and $p\equiv -d \pmod{D}$, to $\T_2$'s root and then rooting $\T_2$ at the new endpoint of this path. 
After this modification, the sum of the depths of the nodes where two leaves with the same label are attached, with $d'$ being the depth of the node at $\T_1$, is congruent to $d'+(d+1-d')+p\equiv 1 \pmod{D}$.
Any reference to $\T_2$ in the remainder of this proof refers to $\T_2$ after this modification.

Let us now describe how to partition $\T_{1}$ and $\T_{2}$ into fragments.
We select all nodes at depth~$\delta$ in both $\T_{1}$ and $\T_2$, such that $\delta \equiv 0 \pmod{D}$.
For every selected node $w$, we create two copies of it, $w'$ and $w''$, such that $w'$ inherits the ingoing edge from $w$'s parent and becomes a leaf of a fragment, while $w''$ inherits the outgoing edges to children of $w$ and becomes the root of some fragment.
We will say that the new leaf $w'$ corresponds to the original $w$, and every node that was not split corresponds to itself.
We conceptually add extra nodes to the fragments with depth smaller than $D$ so as to make all of them identical. (One can think of this as having an embedding of each fragment to a complete $b$-ary tree of depth $D$.)

Consider a fragment $A$ of $\T_{1}$ and a fragment $B$ of $\T_{2}$. We define the corresponding instance of \textsc{Restricted HIA},
denoted $(A,B)$, as follows. Consider a pair of leaves $u$ and $v$ with the same unique label attached to the nodes corresponding
to $u'\in A$ and $v'\in B$ in the original instance. We attach a pair of such leaves to $u'$ and $v'$ in $(A,B)$. We claim
that $(A,B)$ is a valid instance of \textsc{Restricted HIA} over $(b,m)$-normal
trees, where $m=n^{1/c}$.
This requires checking that the number of leaves attached to a node at depth $d'$ in $A$ or $B$ is at most $b^{D+1-d'}$.
In the original instance, a pair of leaves is attached at depths $d'$ and $d+1-d'$, for some $d'=1,2,\ldots,d$.
In $(A,B)$ the corresponding pair of leaves is attached at depths $x$ and $y$ with $x+y \equiv 1 \pmod{D}$ as argued above. 
Moreover, recall that the roots of fragments do not correspond to nodes of the original trees and hence $0 <x \leq D$ and $0 <y \leq D$.
Consequently, $x+y=D+1$.
For any $d'=1,2,\ldots,D$ and any node $u$ at depth $d'$ in $A$, we have at most $b^{D+1-d'}$
nodes at depth $D+1-d'$ in $B$, and hence attach at most that many leaves to $u$, so indeed
the number of leaves attached to any node of $A$ is as required.
A symmetric argument can be used to bound the number of leaves attached to any node of $B$, so $(A,B)$ is indeed a valid instance.

We now claim that answering a query in the original instance reduces to $c$ queries in the smaller instances.
In the original instance, we seek an ancestor $u'$ of a leaf $u\in \T_{1}$ and an ancestor $v'$ of
a leaf $v\in \T_{2}$ such that there is a pair of leaves with the same unique label attached to both $u'$ and $v'$.
We iterate over every fragment $A$ above $u$, locating the unique fragment $B$
above $v$ for which it is possible
that $A$ contains a node corresponding to node $g'$ in~$\T_1$ and $B$ contains
a node corresponding to node $f'$ in $\T_2$ such that there is a pair of leaves with the same unique label attached to both $g'$ and $f'$, and querying
the smaller instance $(A,B)$.
Note that such $B$ is unique by the choice of selected nodes and the fact that if $g'$ is at depth $d'$ in $\T_{1}$ then $f'$ must be at depth $d+1-d'+p$ in $\T_{2}$ and $d+1-d'+p=1-d' \pmod{D}$. 

We cannot afford to build a separate structure for every smaller instance $(A,B)$. However, we can afford
to build a single instance of \textsc{Activation HIA} corresponding to a pair of $(b,m)$-normal trees. By our assumption,
such a blank structure takes $\cO(m^c)=\cO(n)$ space. We first gather, for every smaller instance $(A,B)$,
all pairs of leaves that should be attached to $A$'s and $B$'s nodes. This can be done efficiently by iterating over the at most $d\cdot b^{d+1}$ pairs of
leaves with the same unique label in the original instance and maintaining a hash table with pointers to lists corresponding
to the smaller instances. Then, for every smaller instance, we proceed as follows. We first issue updates to the blank
structure to make it correspond to the current smaller instance. During the updates, we save the modified memory
locations, and after the final update we prepare a hash table mapping the modified memory location to its final content.
Then, we reset the blank structure to its initial state. This takes $\cOtilde(n\cdot t_{u})$ time and space overall
and allows us to later simulate a query on any smaller instance using the stored blank structure and the appropriate
hash table describing the modified memory locations in $\cO(t_{q})$ time.
\end{proof}

\begin{lemma}
If there exists a polynomial-size structure for maintaining the LCS of two dynamic strings of length $\cO(n^{2}\log^{2} n)$, requiring
$t$ time per update, then there exists a polynomial-size structure for \textsc{Activation HIA}
which processes updates and answers queries in $\cO(t)$ time.
\label{lem:lcs}
\end{lemma}
\begin{proof}
Recall that in \textsc{Activation HIA} we are working with two $(b,n)$-normal trees $\T_{1}$ and $\T_{2}$ obtained
from complete $b$-ary trees of depth $d=\log_{b}n$. 
For every $d'=1,2,\ldots,d$ we consider every pair of leaves $u\in \T_{1}$ and $v\in \T_{2}$ with the same label
attached to $u'\in \T_{1}$ and $v'\in \T_{2}$
at depths $d'$ and $d+1-d'$, respectively, and construct the following gadget:
\[ \pathtree(u') \texttt{\$}  \pathtree^{r}(v'). \]
We concatenate all such gadgets, separated by $\texttt{.}$ characters, to obtain $T$. The length of $T$ is $\cO(n^{2}\log^{2}n)$.

For each pair of leaves $u\in \T_{1}$ and $v\in \T_{2}$ at depth $d$, we create a similar gadget:
\[ \pathtree(u)  \texttt{\&}  \pathtree^{r}(v). \]
We concatenate all such gadgets, separated by $\texttt{,}$ characters, to obtain $S$. The length of $S$ is $\cO(n^{2}\log n)$.

Activating a pair of leaves $u\in \T_{1}$ and $v\in \T_{2}$ is implemented with replacing $\$$ by $\#$ in the corresponding
gadget. To answer a query concerning a pair of leaves $u\in \T_{1}$ and $v\in \T_{2}$, we temporarily replace
$\$$ by $\#$ in the corresponding gadget, find the length $L$ of the LCS, and then restore the gadget.
We claim that $L=d+2$ if and only if the total weight of HIA is $d+1$. Recall that the total weight of HIA is $d+1$
if and only if there exists an ancestor $u'$ of $u$ at depth $d'$ and $v'$ of $v$ at depth $d+1-d'$ with
a pair of leaves having the same unique label. Since gadgets are separated with different characters in $T$ and $S$,
any $\pathtree(w)$ is of length at most $d$, $L=d+2$ must correspond to a substring $\pathtree(u') \# \pathtree(v')$,
for some ancestor $u'$ of $u$ at depth~$d'$ and $v'$ of $v$ at depth $d+1-d'$. But such a substring occurs in $T$
if and only if a pair of leaves connected to $u'$ and $v'$ has been activated, so our answer is indeed correct.
\end{proof}

We are now ready to prove the main result of this subsection.

\begin{theorem}\label{thm:lb}
Any polynomial-size structure for maintaining an LCS of two dynamic strings, each of length at most $n$, requires $\Omega(\log n/ \log \log n)$ time per update operation.
\end{theorem}
\begin{proof}
Theorem~\ref{thm:butterfly} and Lemma~\ref{lem:restricted} imply that any structure of size $\cOtilde(n)$
for \textsc{Restricted HIA} requires query time $\Omega(\log n/\log\log n)$. By combining this with Lemma~\ref{lem:activation}
we get that any polynomial-size structure for \textsc{Activation HIA} with polylogarithmic query time requires update
time $\Omega(\log n/\log\log n)$. By combining this with Lemma~\ref{lem:lcs} we get the claimed result.
\end{proof}

\paragraph{Lower bound for amortized update time.}
Let us assume towards a contradiction that there exists an $\cO(n^k)$-size data structure, for constant $k$,  with $t=o(\log n / \log \log n)$ amortized update time 
for the fully dynamic LCS problem over two strings of length at most $n$.

Let $m=n^{1/2k}$.
We initialize the aforementioned fully dynamic LCS data structure for two strings $S$ and $T$, each of length $\cOtilde(m^2)$.
We follow the proof of~\cref{lem:lcs}, performing $\cOtilde(m^2)$ updates to $S$ and $T$ in order to encode 
the \textsc{Activation HIA} problem for two unlabeled $(b,m)$-normal trees.
These operations cost $\cOtilde(m^2 \cdot t)$ time and leave us with a budget of no more than $\cOtilde(m^2\cdot t)=\cOtilde(n)$ towards future updates.
Let us note at this point that the reduction underlying~\cref{lem:lcs} allows us to consider a slightly more general problem than \textsc{Activation HIA}, 
where deactivating a pair of leaves is also allowed.

Recall that in the proof of~\cref{lem:activation} we wish to perform several sequences $L_1, \ldots, L_r$ of activations in the blank \textsc{Activation HIA} structure,
of total length $\cOtilde(n)$, so that sequence $L_j$ creates an instance $\mathcal{I}_j$ of the \textsc{Restricted HIA} problem.
We process each of the sequences as follows.
We first perform the sequence $L_j$ of activations. This adds at most $|L_j| \cdot t$ credits to our budget.
Then, while there exists an expensive HIA query that requires at least $3t$ time to be answered, we perform it.
(Recall that each HIA query reduces to an update in one of the strings that is then reverted.)
When this is no longer possible, all HIA queries cost $\cO(t)$ worst-case time.
We store a snapshot of the current data structure for instance $\mathcal{I}_j$ by storing the changes in the memory from the blank structure using a hash table.
Finally, we deactivate all leaves that were activated by $L_j$. This, again, adds at most $|L_j| \cdot t$ credits to our budget.
In total, our budget starts with $\cOtilde(n)$ credits, and over processing all sequences $L_j$ receives $\cOtilde(n \cdot t)$ extra credits.
As each expensive HIA query uses at least $t$ credits, we ask at most $\cOtilde(n)$ of them, and thus the extra space is $\cOtilde(n)$.
Hence, we obtain an $\cOtilde(n)$-size data structure with $\cO(t)$ worst-case query time for the \textsc{Restricted HIA} problem, a contradiction.

\paragraph{Allowing Las Vegas randomization.}
Let us suppose that there exists data structure for fully dynamic LCS that processes each update in expected amortized time $t$.
Then, by combining our series of reductions with the above credits-based deamortization technique (with queries being expensive or cheap based on their expected cost), we obtain a distribution over $\cOtilde(n)$-size data structures for the \textsc{Butterfly Reachability} with expected query time $\cO(t)$  (expected here is meant in the sense of~\cref{thm:butterflyrand}). By~\cref{thm:butterflyrand}, we have $t=\Omega(\log n/ \log \log n)$.

\bibliographystyle{alphaurl}
\bibliography{biblio}

\end{document}